\documentclass[11pt]{amsart}
\usepackage[margin=3cm]{geometry}
\usepackage{times}
\newtheorem{theorem}{Theorem}[section]
\newtheorem{proposition}[theorem]{Proposition}
\newtheorem{lemma}[theorem]{Lemma}
\newtheorem{corollary}[theorem]{Corollary}
\theoremstyle{definition}
\newtheorem{definition}[theorem]{Definition}
\newtheorem{remark}[theorem]{Remark}
\numberwithin{equation}{section}
\date{}

\usepackage{setspace}
\newcommand{\D}{\mathsf{D}}

\usepackage{amsrefs,enumerate}

\usepackage{tikz-cd}

\DeclareMathOperator{\gh}{\mathit{gh}}
\DeclareMathOperator{\pa}{\mathit{p}}
\DeclareMathOperator{\ad}{ad}
\DeclareMathOperator{\pr}{pr}
\DeclareMathOperator{\Tot}{Tot}
\DeclareMathOperator{\TW}{TW}
\DeclareMathOperator{\MC}{MC}
\DeclareMathOperator{\Sym}{Sym}
\DeclareMathOperator{\End}{End}

\newcommand{\A}{\mathcal{A}}
\newcommand{\B}{\mathcal{B}}
\newcommand{\F}{\mathcal{F}}
\newcommand{\CC}{\mathcal{C}}
\newcommand{\Z}{\mathbb{Z}}
\newcommand{\N}{\mathbb{N}}
\newcommand{\R}{\mathbb{R}}
\newcommand{\C}{\mathbb{C}}
\newcommand{\CO}{\mathcal{O}}
\newcommand{\p}{\partial}
\newcommand{\half}{\tfrac{1}{2}}
\newcommand{\U}{\mathcal{U}}
\newcommand{\V}{\mathcal{V}}
\newcommand{\W}{\mathcal{W}}
\newcommand{\X}{\mathsf{X}}
\newcommand{\h}{\mathsf{H}}
\renewcommand{\[}{(\!(}
\renewcommand{\]}{)\!)}

\renewcommand{\S}{\mathsf{S}}
\renewcommand{\SS}{\mathbb{S}}

\newcommand{\eps}{\varepsilon}
\newcommand{\om}{\omega}
\newcommand{\Om}{\Omega}
\newcommand{\tint}{{\textstyle\int}}
\renewcommand{\o}{\otimes}
\newcommand{\Cech}{\v{C}ech}
\newcommand{\bull}{\bullet}
\newcommand{\Linf}{$L_\infty$}
\renewcommand{\d}{\mathsf{d}}
\renewcommand{\SS}{\mathbb{S}}
\renewcommand{\ss}{\mathbf{s}}
\renewcommand{\t}{\mathsf{t}}
\newcommand{\s}{\mathsf{s}}
\newcommand{\f}{\mathsf{f}}
\newcommand{\g}{\mathsf{g}}
\renewcommand{\AA}{\mathbf{A}}
\newcommand{\<}{\langle}
\renewcommand{\>}{\rangle}
\renewcommand{\L}{\mathcal{L}}
\newcommand{\OM}{\mathsf{w}}

\title{Covariance in the Batalin--Vilkovisky formalism and the
  Maurer--Cartan equation for curved Lie algebras}

\author{Ezra Getzler}

\date{}

\address{Department of Mathematics, Northwestern University,
    Evanston, Illinois, USA}

\email{getzler@northwestern.edu}

\begin{document}

\begin{abstract}
  \setstretch{1.05}
  We express covariance of the Batalin--Vilkovisky formalism in
  classical mechanics by means of the Maurer--Cartan equation in a
  curved Lie superalgebra, defined using the formal variational
  calculus and Sullivan's Thom--Whitney construction. We use this
  framework to construct a Batalin--Vilkovisky canonical
  transformation identifying the Batalin--Vilkovisky formulation of
  the spinning particle with an AKSZ field theory.
\end{abstract}

\maketitle

\section{Introduction}

In the Batalin--Vilkovisky formalism, a classical field theory is
specified by a solution of the classical master equation
\begin{equation}
  \label{master}
  \half \tint ( S , S ) = 0 .
\end{equation}
Alexandrov et al.\ \cite{AKSZ} have studied a particularly important
class of solutions of this equation, known as AKSZ field theories. An
AKSZ field theory in dimension $d$ is a non-linear sigma-model in
which the target is a graded supermanifold $M$ with a shifted
symplectic form of ghost number $d-1$. There is a function $W$ on $M$
of ghost number $d$ satisfying the Maurer--Cartan equation
\begin{equation*}
  \half \{ W , W \} = 0 .
\end{equation*}
Thus, $W$ determines a Hamiltonian vector field $Q(f)=\{W,f\}$ on $M$
of ghost number $1$ and odd parity, which is cohomological:
\begin{equation*}
  Q^2 = \half [Q,Q] = 0 .
\end{equation*}
Important examples of AKSZ field theories are Chern--Simons theory
(Axelrod and Singer \cite{AS}) and the Poisson sigma-model (Cattaneo
and Felder \cites{CF1,CF2}).

In this paper, we restrict attention to field theories with $d=1$, in
other words, classical mechanics. Our main constructions should have
analogues in all dimensions, but our application, showing that the
particle and spinning particle possess hidden AKSZ field theories,
only requires the formalism in $d=1$, and we will focus our attention
on that case.

In Sections 2 and 3, we recall some needed background results on
curved Lie algebras and the formal variational calculus.

In Section 4, we show that an AKSZ field theory with $d=1$, associated
to a graded supermanifold $M$ and an exact symplectic form $\om=d\nu$,
gives rise to a Maurer--Cartan element for a certain curved Lie
superalgebra: we call such Maurer--Cartan elements (classical)
covariant field theories.

To incorporate covariant field theories with topological terms, where
the symplectic form $\om$ is no longer exact, we introduce the
Thom--Whitney totalization for cosimplicial curved Lie superalgebras
in Section 5. The Thom--Whitney totalization replaces the rather rigid
homotopies of piecewise linear topology with the more flexible
homotopies of de~Rham theory. In this setting, we associate a
covariant field theory to a graded supermanifold $M$ together with the
following data: a symplectic form $\om\in\Om^2(M)$, a cover
$\U=\{U_\alpha\}$ of $M$, and one-forms $\nu_\alpha\in\Om^1(U_\alpha)$
such that $d\nu_\alpha=\om$.

Topological terms of this type do not occur in AKSZ models when $d>1$,
since the target symplectic form has ghost number $d-1$, and hence is
exact. In fact, our motivation for introducing the Thom--Whitney
totalization in Section 4 is the hope that using it, the superstring
may be understood as a generalized AKSZ model, in the sense that it
extends to a covariant field theory. In \cite{superparticle}, we study
this problem in the setting of the toy model of the superparticle
(though admittedly still with $d=1$).

After a Batalin--Vilkovisky canonical transformation, the covariant
field theory for a particle moving in a curved spacetime with the AKSZ
model introduced in \cite{curved}. In the introduction, we explain
this in the special case of a particle moving in a flat
background. The Lagrangian of this theory is as follows:
\begin{equation}
  \label{S0}
  S_0 = p_\mu \p x^\mu - \half \eta^{\mu\nu} e p_\mu p_\nu .
\end{equation}
The fields $( x^\mu , p_\mu )_{1\le\mu\le n}$ are the coordinates of
the flat space in which the particle moves, and their conjugate
momenta. The remaining (non-propagating) field of the theory is the
graviton, a nowhere-vanishing one-form $e$ on the world-line.

The solution to the Batalin--Vilkovisky master equation for the
particle, extending the Lagrangian \eqref{S0}, incorporates an
additional field, the ghost $c$. This is a fermionic field of ghost
number one, transforming as a world-line vector field, and is
associated to the covariance of the theory under diffeomorphism of the
world-line. The corresponding antifield $c^+$ is a bosonic field of
ghost number $-2$ transforming as a world-line quadratic
differential. Introduce the expression
\begin{equation*}
  \D = x^+_\mu \p x^\mu + p^{+\mu} \p p_\mu - e \p e^+ + c^+ \p c ,
\end{equation*}
of ghost number $-1$. Consider the Lagrangian $S=S_0+S_1$, where
$S_1=c\D$. It is straightforward to check that the action $\int S\,dt$
satisfies the classical master equation \eqref{master}. Form the
graded Lie superalgebra of polynomials in a variable $u$ of degree $2$
with coefficients in the Batalin--Vilkovisky graded Lie
superalgebra. The element $u\tint\D$ lies in the centre of this graded
Lie algebra, and we may form the curved Lie superalgebra with the same
underlying graded Lie algebra but with nonzero curvature
$u\tint\D$. Then
\begin{equation}
  \label{mc-particle}
  \tint \S = \tint ( S + uc^+ )
\end{equation}
is a Maurer--Cartan element in this curved Lie algebra, in other
words, a covariant field theory. This means that it satisfies the
perturbation of the classical master equation \eqref{master}
\begin{equation*}
  \half \tint ( S , S ) = - u\tint\D .
\end{equation*}

This field theory bears some resemblance to a Chern--Simons field
theory. Recall that the Batalin--Vilkovisky
extension of the Lagrangian for Chern-Simons theory (Axelrod and
Singer \cite{AS}) may be expressed in terms of a composite field
\begin{equation*}
  \AA = c + A + A^+ + c^+ .
\end{equation*}
Here, $A$ is the Chern-Simons field, a connection form on the
3-manifold $M$ for the Lie algebra $\g$, $c\in\Om^0(M,\g)$ is the
ghost field for local gauge transformations, and $A^+\in\Om^2(M,\g)$
and $c^+\in\Om^3(M,\g)$ are their respective antifields. The top
degree component of the differential form
\begin{equation*}
  \half \< \AA , d\AA \> + \tfrac16 \< \AA , [\AA,\AA] \> ,
\end{equation*}
is the Batalin--Vilkovisky Lagrangian $S=S_0+S_1$ of the Chern-Simons
theory:
\begin{align*}
  S_0 &= \half \< A , dA \> + \tfrac16 \< A , [A,A] \>
  &
  S_1 &= \< A^+ , dc +  [A,c] \> + \half \< c^+ , [c,c] \> .
\end{align*}

In order to see that the complete action $\int S$ of the particle has
a hidden AKSZ structure, we apply to it a sequence of canonical
transformations.  Consider the flow $\Phi_\tau$ associated to the
Hamiltonian $c x^+_\mu p^{+\mu}$: this is the solution to the ordinary
differential equation
\begin{equation*}
  \frac{d(\Phi_\tau^*f)}{d\tau} = ( cx^+_\mu p^{+\mu} , f ) ,
\end{equation*}
and is given by the explicit formula
\begin{align*}
  \Phi_\tau^*(x^\mu)
  &= x^\mu + \tau c p^{+\mu}
  &
  \Phi_\tau^*(x^+_\mu)
  &= x^+_\mu \\
  \Phi_\tau^*(p_\mu)
  &= p_\mu - \tau c x^+_\mu
  &
 \Phi_\tau^*(p^{+\mu})
  &= p^{+\mu} \\
  \Phi_\tau^*(e)
  &= e
  &
  \Phi_\tau^*(e^+)
  &= e^+ \\
  \Phi_\tau^*(c)
  &= c
  &
    \Phi_\tau^*(c^+)
  &= c^+ + \tau x^+_\mu p^{+\mu}
\end{align*}
Under this flow, the densities $S_0$ and $S_1$ transform as follows:
\begin{align*}
  \Phi_\tau^*\S_0
  &= S_0 - \tau c ( x_\mu^+ \p x^\mu + p^{+\mu} \p p_\mu ) + \tau \p(c
    p^{+\mu}p_\mu) + \tau^2 c\p c x_\mu^+p^{+\mu} +
    \tau\eta^{\mu\nu} e c p_\mu x^+_\nu  \\
  \Phi_\tau^*\S_1
  &= S_1 - \tau c\p c x^+_\mu p^{+\mu} .
\end{align*}
Let $\Phi$ be the canonical transformation obtained by evaluating the
flow $\Phi_\tau$ at $\tau=1$: we see that
\begin{equation*}
  \Phi^*S = S_0 + ec ( \eta^{\mu\nu} p_\mu x^+_\nu - \p e^+ ) +
  c^+c\p c + \p(c p_\mu p^{+\mu}) .
\end{equation*}

Next, consider the canonical transformation $\Psi$ which leaves the
fields $x^\mu$ and $p_\mu$ and their antifields fixed, and acts on the
remaining fields by the formulas
\begin{align*}
  \Psi^*e &= e & \Psi^*e^+ &= e^+ + e^{-1} c^+c &
  \Psi^*c &= e^{-1}c & \Psi^*c^+ &= ec^+ .
\end{align*}
Formally, this is the value of the flow $\Psi_t$ generated by the
Hamiltonian $\log(e)c^+c$ at $\tau=1$. The canonical transformation
$\Xi=\Phi\circ\Psi$ obtained by composing $\Phi$ and $\Psi$ transforms
the complete Lagrangian $S$ as follows:
\begin{equation*}
  \Xi^*S = \Psi^*\Phi^*S = S_0 + c \bigl( \eta^{\mu\nu} p_\mu
  x^+_\nu - \p e^+ \bigr) + \p \bigl( c ( p_\mu p^{+\mu} + ee^+ )
  \bigr) .
\end{equation*}
After this transformation, the Maurer--Cartan element
\eqref{mc-particle} becomes
\begin{equation*}
  \tint \Xi^*\S = \tint \bigl( p_\mu \p x^\mu - \half \eta^{\mu\nu} e
  p_\mu p_\nu + c \bigl( \eta^{\mu\nu} p_\mu x^+_\nu - \p e^+ \bigr) +
  u ( x^+_\mu p^{+\mu} + e c^+ ) \bigr) .
\end{equation*}
In terms of the composite fields
\begin{align*}
  \mathbf{x}^\mu
  &= x^\mu+dt\,p^{+\mu}
  &
    \mathbf{p}_\mu
  &= p_\mu - dt\,x^+_\mu \\
  \mathbf{c}
  &= c-dt\,e
  &
    \mathbf{b}
  &= e^++dt\,c^+
\end{align*}
we see that $\Xi^*S$ equals the coefficient of $dt$ in the
differential form
\begin{equation*}
  \mathbf{p}_\mu d\mathbf{x}^\mu +
  \mathbf{c} d\mathbf{b} + \half \eta^{\mu\nu} \mathbf{c}
  \mathbf{p}_\mu \mathbf{p}_\nu ,
\end{equation*}
modulo total derivatives. In summary, the particle embeds, by an
explicit canonical transformation, in an AKSZ field theory with fields
$\{\mathbf{x}^\mu,\mathbf{p}_\mu,\mathbf{c},\mathbf{b}\}$. A similar
transformation, for three-dimensional gravity, has been studied by
Cattaneo, Schiavina and Selliah \cite{CSS}.

In passing, we note that the canonical transformation $\Xi$ is the
value at $\tau=1$ of the flow $\Xi_\tau$ associated to the
Batalin--Vilkovisky Hamiltonian
\begin{equation}
  \label{Xi}
  \frac{\log(e)}{e-1} \, c \bigl( ( x^+_\mu p^{+\mu} + ec^+ ) - c^+
  \bigr) .
\end{equation}

In Section~6, we will show that the above remarks may be generalized
to a general covariant field theory coupled to the gravity multiplet
$(e,c)$. In particular, this includes the case of a particle in a
curved spacetime with a background electromagnetic field. 

In Section~7, we turn to the spinning particle, which we have
previously studied in the Batalin--Vilkovisky formalism
\cite{curved}. The spinning particle is a toy model for a
supersymmetric sigma-model coupled to supergravity, in which the
worldline (or spacetime) is reduced from two to one dimensions. (The
corresponding quantum system has Hamiltonian the square of the Dirac
operator.) The fields of this model, in addition to
$\{x^\mu,p_\mu,e,c\}$, comprise fermionic fields $\psi^\mu$ and $\chi$
and the bosonic ghost $\gamma$, supersymmetric partners to $x^\mu$,
$e$ and $c$ respectively. In a flat background, the Lagrangian of the
spinning particle equals
\begin{equation*}
  S_0 = p\p x + \half \psi\p\psi - \half e p^2 + \chi p\psi ,
\end{equation*}
and the associated solution to the classical master equation is
\begin{align}
  \label{Spsi}
  S &= S_0 + c(x^+\p x+p^+\p p+\psi^+\p\psi -e\p e^++c^+\p c
      - \chi\p\chi^++\gamma^+\p\gamma) \notag \\
    &\quad - \gamma ( \p\chi^+ - p\psi^+ + \psi x^+ + 2\chi e^+ ) \\
    &\quad +
      e^{-1}\gamma^2(c^+-x^+p^+-\half\psi^+\psi^+-\chi\gamma^+) .
      \notag
\end{align}
Let $\Xi_\tau$ be the flow associated to the Batalin--Vilkovisky
Hamiltonian
\begin{equation*}
  \frac{\log(e)}{e-1} \, c \bigl( ( x^+_\mu p^{+\mu} + \half
  \eta^{\mu\nu} \psi^+_\mu \psi^+_\nu + ec^+ + \chi\gamma^+ ) - c^+
  \bigr)
\end{equation*}
generalizing \eqref{Xi}, and let $\Xi=\Xi_{\tau=1}$ be the value of
the flow at $\tau=1$. After transformation by the Batalin--Vilkovisky
canonical transformation $\Xi$, the Lagrangian \eqref{Spsi} becomes
the AKSZ field theory
\begin{equation*}
  \Xi^*S = S_0 - c( \p e^+ - px^+ ) - \gamma ( \p\chi^+ - p\psi^+ + \psi
  x^+ + 2\chi e^+ ) + \gamma^2c^+ .
\end{equation*}
In addition to the previous composite fields
$(\mathbf{x},\mathbf{p},\mathbf{c},\mathbf{b})$ associated to the
particle, we now have the additional composite fields
\begin{align*}
  \boldsymbol{\psi}^\mu
  &= \psi^\mu+dt\,\eta^{\mu\nu} \psi^+_\nu
  &
  \boldsymbol{\gamma}
  &= - \gamma + dt\,\chi
  &
    \boldsymbol{\beta}
  &= \chi^++dt\,\gamma^+ .
\end{align*}
Up to a total derivative, the Lagrangian $\Xi^*S$ equals the
coefficient of $dt$ in the expression
\begin{equation*}
  \mathbf{p}_\mu d\mathbf{x}^\mu - \half \eta_{\mu\nu}
  \boldsymbol{\psi}^\mu d\boldsymbol{\psi}^\nu + \mathbf{c} d\mathbf{b}
  + \boldsymbol{\gamma}d\boldsymbol{\beta} + \half \eta^{\mu\nu} \mathbf{c}
  \mathbf{p}_\mu \mathbf{p}_\nu + \boldsymbol{\gamma}
  \mathbf{p}_\mu \boldsymbol{\psi}^\mu + \mathbf{b} \boldsymbol{\gamma}^2 .
\end{equation*}
In Section 7, we carry out the above construction for a more general
class of covariant field theories coupled to the supergravity multiplet
$(e,c,\chi,\gamma)$: this class of theories includes the spinning
particle with curved target, as in \cite{curved}.

\section{The Maurer--Cartan equation in a curved Lie superalgebra}

Let $L^\bull$ be a $\Z$-graded superspace: adopting the language of
theoretical physics, we say that an element $x\in L^k$ has ghost
number $k$, and write $\gh(x)=k$. Furthermore, $L^k$ has a
$\Z/2$-grading, making it into a superspace: we call this grading the
parity, and write $\pa(x)\in\{0,1\}$. We will also say that $x$ is
even (respectively odd) if $\pa(x)=0$ (respectively $1$). A graded
vector space is a special case of a graded superspace, in which the
ghost number and parity are congruent modulo $2$.

A 1-shifted curved Lie superalgebra is a graded superspace with the
following data (all of which have odd parity):
\begin{enumerate}[1)]
\item an element $R\in L^1$ (the curvature);
\item a linear operation $d:L^k\to L^{k+1}$ (the differential);
\item a bilinear operation $(-,-):L^k\times L^\ell\to L^{k+\ell+1}$
  (the antibracket).
\end{enumerate}
The axioms are as follows:
\begin{enumerate}[a)]
\item (the Bianchi identity) $dR=0$;
\item (the curvature identity) for all $x\in L^\bull$,
  \begin{equation*}
    d^2x = (R,x) ;
  \end{equation*}
\item (the Leibniz identity) for all $x,y\in L^\bull$,
  \begin{equation*}
    d(x,y)=(dx,y)+(-1)^{\pa(x)+1}(x,dy) ;
  \end{equation*}
\item (antisymmetry) for all $x,y\in L^\bull$,
  \begin{equation*}
    (y,x) = - (-1)^{(\pa(x)+1)(\pa(y)+1)} (x,y) .
  \end{equation*}
\item (the Jacobi rule) for all $x,y,z\in L^\bull$,
  \begin{equation*}
    (x,(y,z))=((x,y),z)+(-1)^{\pa(x)+1}(y,(x,z)) .
  \end{equation*}
\end{enumerate}
All curved Lie superalgebras considered in this paper are 1-shifted.

Let $L^\bull$ be a curved Lie superalgebra. If $x\in L^k$, we denote
the operation
\begin{equation*}
  y \mapsto (x,y) : L^\bull \to L^{\bull+k+1}
\end{equation*}
by $\ad(x)$. A curved Lie superalgebra $L$ is nilpotent if, for every
odd element $x\in L^\bull$, the endomorphism $\ad(x)$ is nilpotent.

A Maurer--Cartan element in a curved Lie algebra is an even element
$x\in L^0$ such that the following equation holds:
\begin{equation*}
  R + dx + \half (x,x) = 0 .
\end{equation*}
The set of all Maurer--Cartan elements is denoted $\MC(L)$. The
importance of Maurer--Cartan elements stems from the following result.
\begin{lemma}
  If $x\in\MC(L)$, the operator $d+\ad(x):L^\bull\to L^{\bull+1}$ is a
  differential (a graded derivation of square zero).
\end{lemma}
\begin{proof}
  It is evident that $d+\ad(x)$ is a graded derivation. Moreover, we
  have
  \begin{align*}
    (d+\ad(x))^2y &= (d+\ad(x))(dy+(x,y)) \\
                  &= d^2y + d(x,y) + (x,dy) + (x,(x,y)) \\
                  &= (R,y) + (dx,y) + (x,(x,y)) .
  \end{align*}
  The proof is completed by observing that
  $(x,(x,y)) = \half ((x,x),y)$.
\end{proof}

In the special case in which the curvature is zero, we recover the
definition of Maurer--Cartan elements in a differential graded Lie
superalgebra.

If $L$ is a curved Lie superalgebra, the space of odd elements of
$L^{-1}$ form a Lie algebra. If $L$ is nilpotent, there is a gauge
action of this Lie algebra on the set of Maurer--Cartan elements,
given by the equation
\begin{equation*}
  x\bull y  = x + \sum_{n=0}^\infty \frac{(-\ad(y))^n(dy+(x,y))}{(n+1)!} .
\end{equation*}
Informally, this formula expresses the conjugation of the differential
$d+\ad(x)$ by the gauge transformation $e^{\ad(y)}$
\begin{equation*}
  d + \ad(x\bull y) = e^{-\ad(y)} \circ ( d + \ad(x) ) \circ e^{\ad(y)} .
\end{equation*}
This explains why the action preserves solutions of the Maurer--Cartan
equation. In particular, if $dy=0$, then $x\bull y=e^{-\ad(y)}x$.

In order to derive the formula for $x\bull y$, one introduces a
parameter $s$, and considers the ordinary differential equation
\begin{align*}
  \frac{d\ad(x\bull sy)}{ds} &= \frac{d\ }{ds} e^{-\ad(sy)} ( d + \ad(x) )
                            e^{\ad(sy)} \\
                          &= [ d + \ad(x\bull sy) , \ad(y) ] \\
                          &= \ad(dy + ( x\bull sy , y ) ) .
\end{align*}
This leads to the consideration of the ordinary differential equation
\begin{equation}
  \label{ode}
  \frac{d(x\bull sy)}{ds} = dy + ( x\bull sy , y ) ,
\end{equation}
with initial condition $x\bull sy=x$ at $s=0$, whose solution is
\begin{equation*}
  x\bull sy  = x + s \sum_{n=0}^\infty
  \frac{(-s\ad(y))^n(dy+(x,y))}{(n+1)!} .
\end{equation*}

The Baker--Campbell--Hausdorff formula gives an expression for the
composition of two gauge transformations. For a proof, see Tao
\cite{Tao}*{Section~1.2}.
\begin{proposition} 
  \label{BCH} 
  If $y,z$ are odd elements of $L^{-1}$ and $x\in\MC(L)$, we have
  $x\bull y\bull z = x\bull(y\ast z)$, where
  \begin{equation*}
    y \ast z = y + \int_0^1 \frac{\ad(y*tz)}{1-e^{-\ad(y*tz)}} z \, dt
  \end{equation*}
  is the solution of the equation
  $e^{\ad(y\ast z)}=e^{\ad(y)}e^{\ad(z)}$.
\end{proposition}

\section{Formal variational calculus and the classical
  Batalin--Vilkovisky master equation}

Let $M$ be a graded supermanifold, with coordinates
$\{\xi^a\}_{a\in A}$, where $\xi^a$ has ghost number $\gh(\xi^a)\in\Z$
and parity $\pa(\xi^a)\in\Z/2$. Introduce the shifted cotangent bundle
$T^*[1]M$, whose coordinates are the coordinates $\{\xi^a\}_{a\in A}$
of $M$, and dual coordinates $\{\xi_a^+\}_{a\in A}$, of ghost number
\begin{equation*}
  \gh(\xi_a^+) = \gh(\xi^a) - 1 ,
\end{equation*}
and parity
\begin{equation*}
  \pa(\xi_a^+) = 1-\pa(\xi^a) .
\end{equation*}
In the Batalin--Vilkovisky formalism, the coordinates $\xi^a$ are
called fields, and the coordinates $\xi^+_a$ are called
antifields. However, this division is somewhat arbitrary, since we may
just as well exchange the r\^oles of field $\xi^a$ and antifield
$\xi^+_a$. In the work of Batalin and Vilkovisky, it was assumed that
the fields have nonnegative ghost number and the antifields have
negative ghost number, but this proves to be too restrictive in the
setting of AKSZ field theories.

Let $\CO(M)$ be the graded commutative superalgebra of functions on
$M$ (which may be polynomial, rational, analytic, or differentiable,
depending on the setting). Let $\CO_\infty(M)$ be the graded
superspace of all differential expressions in the fields and
antifields, graded by total ghost number, that is, (graded)
polynomials over $\CO(M)$ in the formal derivatives
$\{\p^k\xi^a\}_{k>0}$ of the coordinates with respect to a formal
parameter $t$. In other words, $\CO_\infty(M)$ is the graded
commutative superalgebra of functions on the jet space $J_\infty M$ of
$M$. Note that for now we only consider expressions that carry no
explicit dependence on the variable $t$.

Let $\A(M)$ be the graded commutative superalgebra generated over
$\CO_\infty(M)$ by (graded) polynomials in the derivatives
$\{\p^k\xi^+_a\}_{k\ge0}$ of the antifields. (In fact, one should take
a certain completion of this algebra whereby we allow infinite sums of
terms with decreasing ghost number, but we will be sloppy and neglect
this subtlety here, as we did in \cites{cohomology,curved}. Working
with this completion would not affect the conclusions of those
papers.) This is the graded commutative superalgebra of functions on
the jet space $J_\infty T^*[1]M$.

Introduce the abbreviations
\begin{align*}
  \p_{k,a}
  &= \frac{\p~}{\p(\p^k\xi^a)} : \A^j\to\A^{j-\gh(\xi^a)} , &
  \p^a_k
  &= \frac{\p~}{\p(\p^k\xi^+_a)} : \A^j\to\A^{j-\gh(\xi^+_a)} .
\end{align*}
Let $\p$ be the total derivative with respect to $t$:
\begin{equation*}
  \p = \sum_{k=0}^\infty \bigl( (\p^{k+1}\xi^a) \p_{k,a} +
  (\p^{k+1}\xi^+_a) \p^a_k \bigr) .
\end{equation*}

Let $\phi:M_0\to M_1$ be an \'etale map (local embedding) of graded
supermanifolds, where $M_0$ has coordinates $\{\xi^a\}_{a\in A}$ and
$M_1$ has coordinates $\{\eta^b\}_{b\in B}$: such a map is determined
by functions
\begin{equation*}
  y^b(\xi) \in \CO(M_0)
\end{equation*}
such that $\phi^*\eta^b=y^b(\xi)$. This defines a morphism of algebras
$\phi^*:\CO(M_1)\to\CO(M_0)$, which extends to a morphism
\begin{equation}
  \label{phi}
  \phi^* : \CO_\infty(M_1) \to \CO_\infty(M_0)
\end{equation}
by the requirement that $\p\phi^*=\phi^*\p$, so that
\begin{equation*}
  \phi^*\p^k\eta^b = \p^k y^b(\xi) .
\end{equation*}
In particular,
\begin{equation*}
  \phi^*\p\eta^b = J(\xi)^b_a \p\xi^a ,
\end{equation*}
where $J(\xi)^b_a$ is the Jacobian of $\phi$,
\begin{equation*}
  J(\xi)^b_a = \frac{\p y^b(\xi)}{\p\xi^a} .
\end{equation*}
Since $\phi$ is \'etale, $J$ is invertible. The morphism \eqref{phi}
extends to a morphism
\begin{equation*}
  \phi^*:\A(M_1)\to\A(M_0) ,
\end{equation*}
on setting $\phi^*\eta^+_b = J^{-1}(\xi)^a_b \xi^+_a$, and
\begin{equation*}
  \phi^*\p^k\eta^+_b = \p^k \bigl( J^{-1}(\xi)^a_b \xi^+_a \bigr) .
\end{equation*}

An evolutionary vector field is a graded derivation of the graded
commutative superalgebra $\A(M)$ that commutes with $\p$. In other
words, it is a vector field of the form
\begin{equation*}
  \pr \bigl( X^a \p_a + X_a \p^a \bigr)
  = \sum_{k=0}^\infty \bigl( (\p^kX^a) \p_{k,a} + (\p^kX_a) \p^a_k \bigr) .
\end{equation*}
The evolutionary vector field associated to the expression
$X^a \p_a + X_a \p^a$ by the above formula is called its prolongation.

The Soloviev antibracket on $\A(M)$ is defined by the formula
\begin{equation}
  \label{Soloviev}
  \[ f , g \] = (-1)^{(\pa(f)+1)\pa(\xi^a)}
  \sum_{k,\ell=0}^\infty \bigl( \p^\ell ( \p_{a,k}f ) \, \p^k
  (\p^a_\ell g) + (-1)^{\pa(f)} \p^\ell ( \p^a_kf ) \, \p^k (
  \p_{a,\ell}g ) \bigr) .
\end{equation}
It satisfies the axioms for a graded Lie superalgebra, is linear over
$\p$,
\begin{equation*}
  \[ \p f , g \] = \[ f , \p g \] = \p \[ f , g \] ,
\end{equation*}
and invariant under \'etale changes of coordinates
\cite{Darboux}*{Theorem 4.1}:
\begin{equation*}
  \[ \phi^*f , \phi^*g \] = \phi^* \[ f , g \] .
\end{equation*}

The superspace $\F=\A/\p\A$ of functionals is the graded quotient of
$\A$ by the subspace $\p\A$ of total derivatives. Denote the image of
$f\in\A$ in $\F$ by $\int f$. The Soloviev antibracket $\[f,g\]$
descends to an antibracket
\begin{equation*}
  \tint (f,g)
\end{equation*}
on $\F$, called the Batalin--Vilkovisky antibracket. Denote by
$\delta_a:\F^j\to\A^{j-\gh(\xi^a)}$ and
$\delta^a:\F^j\to\A^{j-\gh(\xi_a^+)}$ the variational derivatives
\begin{align*}
  \delta_a &= \sum_{k=0}^\infty (-\p)^k \circ \p_{k,a} &
  \delta^a &= \sum_{k=0}^\infty (-\p)^k \circ \p^a_k .
\end{align*}

\begin{lemma}
  The Batalin--Vilkovisky antibracket is given by the formula
  \begin{equation*}
    \tint ( f , g ) =  (-1)^{(\pa(f)+1)\pa(\xi^a)}
    \int \bigl( (\delta_af) \, (\delta^ag) + (-1)^{\pa(f)} (\delta^af)
    \, (\delta_ag) \bigr) .
  \end{equation*}
\end{lemma}

The (Batalin-Vilkovisky) Hamiltonian vector field associated to an
element $\tint f\in\F$ is the evolutionary vector field given by the
formula
\begin{equation*}
  \h_f = \sum_{k=0}^\infty (-1)^{(\pa(f)+1)\pa(\xi^a)}
  \bigl( \p^k (\delta_af) \, \p^a_k + (-1)^{\pa(f)} \p^k (\delta^af)
  \, \p_{k,a} \bigr) .
\end{equation*}
Despite the notation, $\h_f$ only depends on $f$ through its image
$\tint f$ in $\F$.

The following theorem is proved, though not in precisely these terms,
in Olver \cite{Olver}, but we give here a simpler proof, taken from
\cite{curved}.
\begin{theorem}
  \label{hamiltonian}
  The map $f\mapsto\h_f$ is a morphism of graded Lie superalgebras
  from $\F$ to the evolutionary vector fields.
\end{theorem}

Recall the higher Euler operators of Kruskal et al.\ \cite{KMGZ}:
\begin{align*}
  \delta_{k,a}
  &= \sum_{\ell=0}^\infty \tbinom{k+\ell}{k} \, (-\p)^\ell
    \p_{k+\ell,a} ,
  & \delta^a_k
  &= \sum_{\ell=0}^\infty \tbinom{k+\ell}{k} \, (-\p)^\ell \p^a_{k+\ell} .
\end{align*}
When $k=0$, $\delta_{0,a}=\delta_a$ and $\delta^a_0=\delta^a$ are the
classical variational derivatives.

If $f\in\A$, the differential operator $\ad(f)=\[f,-\]$ associated to
$f$ by the Soloviev antibracket is given by the formula
\cite{curved}*{Proposition~2.1}
\begin{equation}
  \label{adf}
  \ad(f) = \sum_{k=0}^\infty \p^k \f_{(k)} ,
\end{equation}
where $\f_{(k)}$ is the sequence of evolutionary vector fields
\begin{equation*}
  \f_{(k)} = (-1)^{(\pa(f)+1)\pa(\xi^a)} \pr\bigl( \bigl( \delta_{k,a}f
  \bigr) \p^a + (-1)^{\pa(f)} \bigl( \delta^a_kf \bigr) \p_a
  \bigr) .
\end{equation*}
In particular, $\f_{(0)}=H_f$.

The proof of Theorem \ref{hamiltonian} relies on
\cite{curved}*{Theorem~2.1}, which we reformulate for convenience.
\begin{lemma}
  \label{recursion}
  Let $\t_k$, $k\ge0$, be a sequence of evolutionary vector fields
  such that $\t_k=0$, $k\gg0$, and
  \begin{equation*}
    \label{ss}
    \sum_{k=0}^\infty \p^k\t_k = 0 . 
  \end{equation*}
  Then $\t_k=0$ for all $k\ge0$.
\end{lemma}
\begin{proof}
  We prove by downward induction in $k$ that the vector fields $\t_k$
  vanish. Let $K$ be the largest integer such that $\t_K$ is
  nonzero. Let $\xi$ be one of the fields of the theory having even
  parity, and take the $(K+1)$-fold commutator of the left-hand side
  of \eqref{ss} with $\xi$. We obtain
  \begin{equation*}
    (K+1)! \, (\p\xi)^K \t_K(\xi) = 0 .
  \end{equation*}
  It follows that $\t_K(\xi)=0$.

  Next, we take the commutator with the antifield $\xi^+$ followed by
  the $K$-fold commutator with $\xi$: we obtain the equation
  \begin{equation*}
    K! \, (\p\xi)^{K-1} \bigl( \p\xi \, \t_K(\xi^+) +
    K\, \p\xi^+ \, \t_K(\xi) \bigr) = K! \, (\p\xi)^K \t_K(\xi^+) = 0 .
  \end{equation*}
  We conclude that $\t_K(\xi^+)=0$.

  The vanishing of $\t_K(\xi)$ and $\t_K(\xi^+)$ may be proved for
  fields $\xi$ of odd parity by exchanging the r\^oles of $\xi$ and
  its antifield $\xi^+$ in the above argument. In this way, we see
  that $\t_K=0$. Arguing by downward induction, we conclude that
  $\t_k=0$ for all $k\ge0$, proving the lemma.
\end{proof}

\begin{proof}[Proof of Theorem~\ref{hamiltonian}]
  If $\[f,g\]=h$, it follows from \eqref{adf} that
  \begin{equation*}
    \sum_{k=0}^\infty \sum_{\ell=0}^k \p^k[\f_{(\ell)},\g_{(k-\ell)}] =
    \sum_{k=0}^\infty \p^k \mathsf{h}_{(k)} .
  \end{equation*}
  Consider the evolutionary vector fields
  \begin{equation*}
    \t_k = \sum_{\ell=0}^k [\f_{(\ell)},\g_{(k-\ell)}] - \mathsf{h}_{(k)} .
  \end{equation*}
  We are in the situation of Lemma~\ref{recursion}: it follows that
  $\t_k=0$ for all $k\ge0$, and in particular,
  \begin{equation*}
    \t_0 = [\h_f,\h_g]-\h_{\[f,g\]} = 0 .
  \end{equation*}
  Since $\int\[f,g\]=\int(f,g)$, we see that the map $\h$ is a
  morphism of graded Lie superalgebras.
\end{proof}

The following lemma shows that the kernel of the Hamiltonian map
$f\mapsto\h_f$  vanishes except in ghost number $0$, where it equals
the constant multiples of $1$.
\begin{theorem}
  \label{faithful}
  If $\h_f=0$, then $f$ is the sum of a constant and a total
  derivative.
\end{theorem}
\begin{proof}[Proof (Olver \cite{Olver}*{Theorem~4.7})]
  We must show that if $\delta_af=\delta^af=0$, then $f$ is the sum of
  a constant and a total derivative. For $0\le s\le 1$, let $f_s$ be
  the rescaled quantity
  \begin{equation*}
    f_s(\xi,\p\xi,\dots) = f(s\xi,s\p\xi,\dots) .
  \end{equation*}
  It is an exercise in binomial coefficients to show that an
  evolutionary vector field may be written in terms of the higher
  Euler operators:
  \begin{equation*}
    \pr \bigl( X^a \p_a + X_a \p^a \bigr)
    = \sum_{k=0}^\infty \p^k \bigl( X^a \delta_{k,a} + X_a \delta^a_k
    \bigr)
  \end{equation*}
  It follows that
  \begin{align*}
    \frac{df_s}{ds}
    &= \pr( \xi^a\,\p_a+\xi^+_a\,\p^a )f_s \\
    &= \sum_{k=0}^\infty \p^k( \xi^a\,\delta_{k,a}f +
      \xi^+_a\,\delta_k^af )_s .
  \end{align*}
  Integrating over $s$ from $0$ to $1$, we see that
  \begin{equation*}
    f = f_0 + \sum_{k=0}^\infty \biggl( \int_0^1 \p^k(
    \xi^a\,\delta_{k,a}f_s + \xi^+_a\,\delta_k^af_s ) \, ds \biggr) .
  \end{equation*}
  In particular, if $\delta_af=\delta^af=0$, we see that
  \begin{equation*}
    f = f_0 + \sum_{k=1}^\infty \p\biggl( \int_0^1 \p^{k-1}(
    \xi^a\,\delta_{k,a}f_s + \xi^+_a\,\delta_k^af_s ) \, ds \biggr) ,
  \end{equation*}
  proving the lemma.
\end{proof}

A Maurer--Cartan element of $\F$ is a solution $\int S\in\F^0$ of the
classical Batalin--Vilkovisky master equation \eqref{master}. In the
Batalin--Vilkovisky formalism, a Maurer--Cartan element $\tint S$
determines a classical field theory.

There is a more precise formulation of the classical master equation,
obtained by lifting a solution in the space of functionals $\F$ to a
resolution of this space. We review the details of this construction,
taken from \cite{curved} .

Introduce the quotient complex $\tilde{\A}$ of $\A$ by the constants:
\begin{equation*}
  \tilde{\A}^j =
  \begin{cases}
    \A^0/\C , & j=0 , \\
    \A^j , & \text{$j\ne0$} .
  \end{cases}
\end{equation*}
The space of functionals $\F$ has a resolution
\begin{equation*}
  \B^j = \A^j \oplus \tilde{\A}^{j+1}\eps ,
\end{equation*}
where the symbol $\eps$ is understood to have odd parity and ghost
number $-1$, so that the parities of the superspace $\tilde{\A}^{j+1}$
are reversed in $\B^j$. The differential $d:\B^j\to\B^{j+1}$ equals
\begin{equation}
  \label{differential}
  d \bigl( f + g\eps \bigr) = (-1)^{\pa(g)} \, \p g ,
\end{equation}
The Soloviev antibracket extends to $\B$ by the formula
\begin{equation}
  \label{Soloviev2}
  \[ f_0 + g_0\eps , f_1 + g_1\eps \] = \[ f_0 , f_1 \] + \[ f_0
  , g_1 \] \, \eps + (-1)^{\pa(f_1)+1} \, \[ g_0 , f_1 \] \, \eps .
\end{equation}
We have
\begin{equation*}
  d \[ a , b\] = \[ da , b \] + (-1)^{\pa(a)+1} \[ a , db \] ,
\end{equation*}
and the differential graded Lie superalgebra $(\B,d)$ is a resolution
of the graded Lie superalgebra $\F$.

If $\tint S$ is a solution of the classical master equation
\eqref{master}, there is an element $\tilde{S}\in\A^1$ of odd parity
such that
\begin{equation*}
  \half \, \[ S , S \] = \p\tilde{S} .
\end{equation*}
The classical master equation \eqref{master} may be recast as the
Maurer-Cartan equation
\begin{equation}
  \label{MASTER}
  d\S + \half \[ \S , \S \] = 0
\end{equation}
in $\B$, where
\begin{equation*}
  \S = S + \tilde{S}\eps \in \B^0 .
\end{equation*}
This refinement of the classical master equation is closely related to
the modified ckassical master equation of Cattaneo, Mn\"ev and
Reshitikhin \cite{CMR}*{Proposition 3.1}.

Let $\s$ be the Hamiltonian vector field $\h_\S$: this is an
evolutionary vector field of degree $1$. By Theorem~\ref{hamiltonian},
we see that a solution of the classical master equation \eqref{MASTER}
yields an odd Hamiltonian vector field $\s$, of ghost number $1$,
satisfying the relation $\s^2=0$. The differential graded Lie
superalgebra $\B$ with differential $d+\s$ is a resolution of the
differential graded Lie superalgebra $\F$, with differential
$\tint(S,-)$. The cohomology of this complex is the
Batalin--Vilkovisky cohomology of the classical field theory $S$.

\section{Covariant field theories in one dimension: local case}

AKSZ field theories are a class of solutions of the classical master
equation, introduced by Alexandrov et al.\ \cite{AKSZ}. Here, we only
consider the case of one-dimensional AKSZ field theories: these
include the main model of interest to us in this paper, the spinning
particle in a curved background. (The focus in \cite{AKSZ} is rather
on the two and three-dimensional cases.) An AKSZ field theory is
associated with a symplectic form $\om$ on the graded supermanifold
whose coordinates are the fields of the theory.

In this section, we define a curved Lie superalgebra whose
Maurer--Cartan elements consist of a solution of the master equation
\eqref{MASTER}, together with additional structure that expresses
covariance with respect to time translation. In the case of an AKSZ
field theory, the additional structure involves the Poisson tensor
$\pi=\om^{-1}$, and thus incorporates the nondegeneracy of the
symplectic form.

Let $u$ be a variable of ghost number $2$. We consider the graded Lie
algebras of power series in $u$ with coefficients in the graded Lie
algebras $\F$ and $\B$, such that
\begin{equation*}
  \tint (uf,g) = \tint (f,ug) = u \tint (f,g) ,
\end{equation*}
respectively
\begin{equation*}
  \[uf,g\] = \[f,ug\] = u\[f,g\] .
\end{equation*}

The element
\begin{equation*}
  \D = \xi^+_a \p\xi^a \in \A(M)
\end{equation*}
is invariant under changes of coordinates, and its image $\tint\D$ in
$\F$ lies in the centre, that is, $\ad(\tint\D)=0$. Consider the
curved Lie superalgebra $\F[[u]]$ with vanishing differential and
curvature $u\tint\D$. A Maurer--Cartan element of $\F[[u]]$ is a
solution $\tint S_u\in\F[[u]]$, of ghost number $0$, of the equation
\begin{equation}
  \label{aksz}
  \half \tint ( S_u , S_u ) = - u\tint\D .
\end{equation}
Expand $S_u$ in powers of $u$
\begin{equation*}
  S_u = S_0 + uS_1 + u^2S_2 + \dots
\end{equation*}
The Maurer--Cartan equation \eqref{aksz} is equivalent to the
classical master equation \eqref{master} for $S=S_0$, the equation
\begin{equation*}
  \tint ( S_0 , S_1 ) = - \tint \D ,
\end{equation*}
and, for $n>1$, the sequence of equations
\begin{equation*}
  \sum_{k=0}^n \tint ( S_k , S_{n-k} ) = 0 .
\end{equation*}
In \cites{cohomology,curved}, such a structure was found in the case
of the spinning particle: in those papers, $S_0$ was called $S$, $S_1$
was called $G$, while $S_n$ vanished for $n>1$.

The operator $\ad(\D):\A^\bull\to\A^\bull$ is given by the explicit
formula
\begin{equation*}
  \ad(\D) = \p\pr( \xi^+_a \p^a ) - \p .
\end{equation*}
Introduce the graded derivation $\iota$ on $\B$, of degree $-1$,
\begin{equation*}
  \iota \bigl( f + g\eps \bigr) = (-1)^{\pa(f)} \bigl( \pr(
  \xi^+_a \p^a )f - f \bigr) \, \eps .
\end{equation*}
It is easily seen that $d\iota + \iota d = \ad(\D)$ and
$\iota^2=0$. Let $\B[[u]]$ be the curved Lie superalgebra with graded
derivation $d_u = d + u\iota$ and curvature $u\D$.

The following definition is  central to this paper.
\begin{definition}
  A (one-dimensional) covariant field theory is a Maurer--Cartan
  element $\S_u\in\B[[u]]$, that is, an element of ghost number $0$ and
  even parity such that
  \begin{equation}
    \label{AKSZ}
    d_u\S_u + \half \, \[ \S_u , \S_u \] = - u\D .
  \end{equation}
\end{definition}

As in the case of the classical master equation, any solution of
\eqref{aksz} gives rise to a solution of \eqref{AKSZ}.
\begin{proposition}
  Let $\tint S_u\in\F[[u]]$ be a solution of \eqref{aksz}, and choose
  a lift of $\tint S_u$ to an element $S_u\in\A[[u]]$. Let
  $\tilde{S}_u\in\tilde{\A}[[u]]$ be the element determined by the
  equation
  \begin{equation*}
    \half \[ S_u , S_u \] = - u\D + \p\tilde{S}_u .
  \end{equation*}
  Then
  \begin{equation*}
    \S_u = S_u + \tilde{S}_u \, \eps \in \B[[u]]
  \end{equation*}
  is a solution of \eqref{AKSZ}.
\end{proposition}
\begin{proof}
  We have
  \begin{align*}
    u\D + d_u\S_u + \half \[ \S_u , \S_u \]
    &= u\D - \p\tilde{S}_u + u\iota(S_u) + \half \[ S_u +
      \tilde{S}_u\eps , S_u + \tilde{S}_u\eps \] \\
    &= u\iota(S_u) + \[ S_u , \tilde{S}_u \] \eps .
  \end{align*}
  We would like to show that the right-hand side has the form a
  constant times $\eps$: it suffices to verify that on applying $\p$ to
  it, we obtain zero. But we have
  \begin{align*}
    u\p\iota(S_u) + \p\[ S_u , \tilde{S}_u \] \eps
    &= \bigl( \[ S_u , -u\D+\p\tilde{S}_u \] \bigr) \eps \\
    &= \half \[ S_u , \[\S_u , \S_u \] \] \eps = 0 .
  \end{align*}
  It follows that $u\iota(S_u) + \[ S_u , \tilde{S}_u \] \eps$
  represents zero in $\B[[u]]$.
\end{proof}

Let $\S_u$ be a covariant field theory. Expanding $\S_u$ in powers of
$u$, we obtain a series of elements $\S_n\in\B^{-2n}$:
\begin{equation*}
  \S_u = \sum_{n=0}^\infty u^n \S_n .
\end{equation*}
Let $\s_n$ be the Hamiltonian vector field $\h_{\S_n}$: this is an
evolutionary vector field of degree $1-2n$. By
Theorem~\ref{hamiltonian}, we see that a covariant field theory yields
a sequence of Hamiltonian vector fields $\s_n$, of ghost number
$1-2n$, such that $\s_0$ is the Batalin--Vilkovisky differential,
satisfying the relation $\s_0^2=0$, $\s_1$ is a cochain homotopy for
the operator $\p$, in the sense that
\begin{equation*}
  [\s_0,\s_1] = -\p ,
\end{equation*}
and for $n>1$,
\begin{equation*}
  \sum_{k=0}^n [\s_k,\s_{n-k}] = 0 .
\end{equation*}
All of the examples considered in this paper satisfy $\S_n=0$, $n>1$;
in particular, $\s_1^2=0$.

An odd element $H\in\B[[u]]$ of ghost number $-1$ generates a flow on
the space of covariant field theories by gauge action on the curved
Lie superalgebra:
\begin{equation*}
  \S_{u,\tau} = \S_u\bull\tau H .
\end{equation*}
We may also consider twists of covariant field theories, by which we
mean the flow associated to a Hamiltonian in
$\B\[u\]=\B[[u]][u^{-1}]$ such that the Maurer--Cartan element
$\S_u\bull\tau H$ remains in $\B[[u]]$. The class of twists discussed
in the following proposition are the ones of importance to the study
of AKSZ field theories.

If $\S_u$ is a covariant field theory, the operator
\begin{equation}
  \label{d}
  \d = d_u+\ad(\S_u)
\end{equation}
is a differential on $\B[[u]]$.
\begin{proposition}
  \label{Aksz}
  Consider an element $W\in\CO(M)$ of ghost number $1$ and odd parity
  such that $\d W$ is divisible by $u$ and $\[\d W,W\]=0$.  Then the
  twist $\S_u\bull u^{-1}W$ of $\S_u$ by $u^{-1}W$ is a covariant
  field theory, given by the formula
  \begin{equation*}
    \S_u\bull u^{-1}W = \S_u + u^{-1} \d W .
  \end{equation*}
\end{proposition}

Let us now show how these formulas capture AKSZ field theories in the
one-dimensional case. The de~Rham complex $\Om^\bull(M)$ of the graded
supermanifold is generated over $\CO(M)$ by the one-forms $df$,
$f\in\CO(M)$, of parity $\pa(f)+1$, subject to the Leibniz relation
\begin{equation*}
  d(fg) = df\,g + (-1)^{\pa(f)} f\,dg .
\end{equation*}
We adopt the sign convention that one-forms graded commute:
\begin{equation*}
  df\,dg = (-1)^{(\pa(f)+1)(\pa(g)+1)} dg\,df .
\end{equation*}

Let
\begin{equation*}
  \nu = \nu_a(\xi) \, d\xi^a \in \Om^1(M)
\end{equation*}
be a one-form on $M$ of ghost number $0$ and odd parity; in other
words, $\gh(\nu_a)=-\gh(\xi^a)$ and $\pa(\nu_a)=\pa(\xi^a)$. The
two-form $\om=d\nu$ equals
\begin{align*}
  \om &= \half \, d\xi^a \, \om_{ab}(\xi) \, d\xi^b \in \Om^2(M) \\
      &= \half (-1)^{(\pa(\xi^a)+1)\pa(\xi^b)} \, \om_{ab}(\xi) \,
        d\xi^a \, d\xi^b ,      
\end{align*}
where
\begin{equation}
  \label{omega}
  \om_{ab} = \p_a\nu_b - (-1)^{\pa(\xi^a)\pa(\xi^b)} \p_b\nu_a .
\end{equation}
In particular, $\gh(\om_{ab})=-\gh(\xi^a)-\gh(\xi^b)$ and
$\pa(\om_{ab})=\pa(\xi^a)+\pa(\xi^b)$. 

Denote the frame of the tangent bundle $TM$ dual to the frame $d\xi^a$
of the cotangent bundle by $\tau_a=\p/\p\xi^a$. The two-form $\om$
induces a morphism of vector bundles $TM\to T^*M$, which is denoted
$X\mapsto X^\flat=X\mathbin\lrcorner\om$, or in terms of the frames
$\{\tau_a\}$ and $\{d\xi^a\}$,
\begin{equation*}
  \tau^\flat_a = \om_{ab}\,d\xi^b .
\end{equation*}
Likewise, a bivector field $\pi$ on $M$ induces a morphism of vector
bundles $T^*M\to TM$, denoted
$\theta\mapsto\theta^\sharp=\pi\mathbin\lrcorner\theta$.  The two-form
$\om$ is symplectic if there is a bivector field $\pi$ such that
$(X^\flat){}^\sharp=X$. Expanding the bivector field in the local
frame $\{\tau_a\}$,
\begin{align*}
  \pi &= \half \, \tau_a \, \pi^{ab}(\xi) \, \tau_b \in
        \Gamma(M,\Sym^2(T[-1]M)) \\
      &= \half \, (-1)^{(\pa(\xi^a)+1)\pa(\xi^b)} \, \pi^{ab}(\xi) \,
        \tau_a \, \tau_b ,
\end{align*}
the relationship between $\om$ and $\tau$ becomes
\begin{equation}
  \label{symplectic}
  (-1)^{\pa(\xi^a)} \, \pi^{ab} \om_{bc}  = \delta_c^a .
\end{equation}
Note that the coefficients $\pi^{ab}$ possess the same symmetry as
$\om_{ab}$, namely
\begin{equation*}
  \pi^{ab} = - (-1)^{\pa(\xi^a)\pa(\xi^b)} \, \pi^{ba} .
\end{equation*}

\begin{lemma}
  If $\om$ and $\pi$ are inverse to each other in the above sense,
  then $\om$ is closed if and only if the bivector field $\pi$ is a
  Poisson tensor:
  \begin{equation}
    \label{Poisson}
    \pi^{ab} \, \p_b\pi^{cd} = (-1)^{\pa(\xi^a)\pa(\xi^c)} \,
    \pi^{cb} \, \p_b\pi^{ad} .
  \end{equation}
\end{lemma}

Let $M$ be a graded supermanifold, and let $\nu$ be a one-form of
ghost number $0$ and odd parity such that $\om=d\nu$ is a symplectic
form. The following theorem follows by a lengthy but straightforward
calculation based on \eqref{omega}, \eqref{symplectic} and
\eqref{Poisson}.
\begin{theorem}
  \label{Su}
  The elements of $\B[[u]]$
  \begin{align*}
    \S_0 &= (-1)^{\pa(\xi^a)} \, \nu_a(\xi) \, \p\xi^a , \\
    \S_1 &= \half \, \xi^+_a \, \pi^{ab}(\xi) \, \xi^+_b
           + (-1)^{\pa(\xi^a)} \, \nu_a(\xi) \, \pi^{ab}(\xi) \,
           \xi^+_b \eps \\
         &= \half ( \xi^+_a - \nu_a \eps ) \pi^{ab} ( \xi^+_b -
           \nu_b \eps ) ,
  \end{align*}
  of ghost number $0$ and $-2$ respectively, and even parity, are
  independent of the coordinate system $\{\xi^a\}$.  Their sum
  \begin{equation*}
    \S_u = \S_0 + u\S_1 \in \B[[u]] ,
  \end{equation*}
  satisfies \eqref{AKSZ}, and hence defines a covariant field theory.
\end{theorem}

Let $\nu$ and $\nu'$ be two one-forms such that
\begin{equation*}
  d\nu = d\nu' = \om ,
\end{equation*}
and in particular, $\nu-\nu'$ is closed. If $M$ is simply connected,
then $\nu-\nu'$ is exact: there is a function $\mu\in\CO(M)$ such that
\begin{equation*}
  \nu - \nu' = d\mu .
\end{equation*}
It follows that
\begin{align*}
  \S_0 - \S_0'
  &= (-1)^{\pa(\xi^a)} \, \p_a \mu(\xi) \, \p\xi^a \\
  &= \p \mu(\xi) ,
\end{align*}
and in particular,
\begin{equation*}
  \tint \S_0 = \tint \S_0'
\end{equation*}
and
\begin{equation*}
  \tint \S_u = \tint \S_u' .
\end{equation*}
Thus, locally in $M$, the choice of $\nu$ is unimportant in the
definition of the field theory: it is only when the worldline has
nonempty boundary (or $M$ has nonzero first homology) that this
ambiguity comes into play. This is one of the reasons that we have
introduced the resolution $\B$ of $\F$

The Poisson bracket associated to the symplectic form $\om$ is the
bilinear form on $\CO(M)$ given by the formula
\begin{align*}
  u \{ f , g \}
  &= (-1)^{\pa(f)} \[ \d f , g \] \\
  &= (-1)^{(\pa(f)+\pa(\xi^a))\pa(\xi^b)} u \, \pi^{ab} \p_af \p_bg ,
\end{align*}
where $\d$ is the differential introduced in \eqref{d}.
\begin{lemma}
  The Poisson bracket satisfies the graded symmetry condition
  \begin{equation*}
    \{ f , g \} = - (-1)^{\pa(f)\pa(g)} \, \{ g , f \} .
  \end{equation*}
\end{lemma}
\begin{proof}
  We have $\[f,g\]=0$, hence
  \begin{align*}
    0 &= \d\[ f , g \] = \[ \d f , g \] +
        (-1)^{\pa(f)+1} \, \[ f , \d g \] \\
      &= \[ \d f , g \] - (-1)^{(\pa(f)+1)(\pa(g)+1)}
        \, \[ \d g , f \] .
        \qedhere
  \end{align*}
\end{proof}

The following lemma generalizes to graded supermanifolds the proof of
the Jacobi rule for the Poisson bracket associated to a Poisson tensor.
\begin{lemma}
  The Poisson bracket satisfies the graded Jacobi identity
  \begin{equation*}
    \{ f , \{ g , h \} \} =  \{ \{ f , g \} , h \} +
    (-1)^{\pa(f)\pa(g)} \, \{ g , \{ f , h \} \} .
  \end{equation*}
\end{lemma}
\begin{proof}
  We have
  \begin{align*}
    u^2 & \{ \{ f , g \} , h \}
          = (-1)^{\pa(g)} \, \[ \d\[ \d f , g \]  , h \] \\
        &= (-1)^{\pa(f)+\pa(g)} \, \[ \[ \d f , \d g \] , h \] \\
        &= (-1)^{\pa(f)+\pa(g)} \bigl( \[ \d f , \[ \d g ,
          h \] \] - (-1)^{\pa(f)\pa(g)} \, \[ \d g , \[ \d f
          , h \] \] \bigr) \\
        &= u^2 \bigl( \{ f , \{ g , h \} \} - (-1)^{\pa(f)\pa(g)} \, \{
          g , \{ f , h \} \} \bigr) .
          \qedhere
  \end{align*}
\end{proof}

If $W\in\CO(M)$ is a function on $M$ of ghost number $1$ and odd
parity, then
\begin{equation*}
  \d W = u \bigl( - \xi^+_a\,\pi^{ab}\,\p_bW + (-1)^{\pa(\xi^a)} \,
  \nu_a \, \pi^{ab}\,\p_bW\eps + W\eps \bigr)
\end{equation*}
is divisible by $u$, and $\[\d W,W\]=u\{W,W\}$ vanishes if and only if
\begin{equation*}
  \{W,W\} = 0 ,
\end{equation*}
in other words, precisely when the Hamiltonian vector field associated
to $W$ is cohomological. Applying Proposition~\ref{Aksz}, we obtain
the following result.
\begin{theorem}
  Let $M$ be a graded supermanifold, and let $\nu\in\Om^1(M)$ be a
  one-form of ghost number $0$ and odd parity such that $\om=d\nu$ is
  a symplectic form. Let $\pi$ be the Poisson tensor associated to
  $\om$.

  Let $W\in\CO(M)$ be a function on $M$ of ghost number $1$ and odd
  parity, such that $\{W,W\}=0$. Then the twist $\S_u\bull u^{-1}W$ of
  the covariant field theory $\S_u$ by $u^{-1}W$, given by the formula
  \begin{align*}
    \S_u + u^{-1} \d W
    &= (-1)^{\pa(\xi^a)} \nu_a\,\p\xi^a + W \eps \\
    &+ \half u ( \xi^+_a - u^{-1}\p_aW - \nu_a \eps )
      \pi^{ab} ( \xi^+_b - u^{-1}\p_bW - \nu_b \eps ) ,
  \end{align*}
  is a covariant field theory.
\end{theorem}

\section{Covariant field theories in one dimension: global case}

The formalism of the last section only applies when the symplectic
form $\om$ on the graded supermanifold $M$ is exact. When this
condition is not satisfied, the best we can do is to choose a cover
\begin{equation*}
  \U =  \{ U_\alpha \}_{\alpha\in I}
\end{equation*}
of $M$, where each $U_\alpha$ is an open subspace of the graded
supermanifold $M$, such that the restriction of $\om$ to each
$U_\alpha$ is exact:
\begin{equation*}
  \om|_{U_\alpha} = d\nu_\alpha .
\end{equation*}

The nerve $N_k\U$ of the cover is the sequence of graded
supermanifolds indexed by $k\ge0$
\begin{equation*}
  N_k\U = \bigsqcup_{\alpha_0\dots\alpha_k\in I^{k+1}}
  U_{\alpha_0\dots\alpha_k} ,
\end{equation*}
where
\begin{equation*}
  U_{\alpha_0\dots\alpha_k} = U_{\alpha_0} \cap \dots \cap U_{\alpha_k} .
\end{equation*}
Denote by $\epsilon:N_0\U\to M$ the map which on each summand
$U_\alpha$ equals the inclusion $U\hookrightarrow M$.

The collection $\nu=\{\nu_\alpha\}_{\alpha\in I}$ gives a one-form on
$N_0\U$, such that
\begin{equation*}
  d\nu = \epsilon^*\om .
\end{equation*}
For all $\alpha_0,\alpha_1\in I$, the one-form
\begin{equation*}
  \nu_{\alpha_0}|_{U_{\alpha_0\alpha_1}}
  - \nu_{\alpha_1}|_{U_{\alpha_0\alpha_1}} \in
  \Om^1(U_{\alpha_0\alpha_1})
\end{equation*}
is closed. Assume the cover $\U$ is chosen such that this form is
exact for all $(\alpha_0,\alpha_1)$: there exists functions
$\mu_{\alpha_0\alpha_1}\in\Om^0(U_{\alpha_0\alpha_1})$ such that
\begin{equation*}
  d\mu_{\alpha_0\alpha_1} = \nu_{\alpha_0}|_{U_{\alpha_0\alpha_1}}
  - \nu_{\alpha_1}|_{U_{\alpha_0\alpha_1}} .
\end{equation*}
Assemble the functions
$\{\mu_{\alpha_0\alpha_1}\}_{\alpha_0\alpha_1\in I}$ into a single
function $\mu$ on $N_1\U$. Let $\delta_0,\delta_1:N_1\U\to N_0\U$ be
the morphisms which on $U_{\alpha_0\alpha_1}$ are respectively the
inclusions into $U_{\alpha_1}$ and $U_{\alpha_0}$ (sic). We have
\begin{equation*}
  d\mu = \delta_1^*\nu - \delta_0^*\nu .
\end{equation*}

The collection of differential forms
\begin{equation*}
  (\nu,\mu) \in \Om^1(N_0\U) \times \Om^0(N_1\U)
\end{equation*}  
will serve as a replacement for the non-existent one-form
$\nu\in\Om^1(M)$ solving the equation $\om=d\nu$. In order to repeat
the discussion of the last section, we must extend the definition of
the Soloviev bracket, and the curved Lie algebra $\B[[u]]$, from
graded supermanifolds $M$ to sequences of graded supermanifolds of the
form $\{N_k\U\}$. Since we will use the formalism of simplicial and
cosimplicial objects in our discussion, we now review their
definition.

Let $\Delta$ be the category whose objects are the totally ordered
sets
\begin{equation*}
  [k] = (0<\dots<k) , \quad k\in\N ,
\end{equation*}
and whose morphisms are the order-preserving functions. A simplicial
graded supermanifold $M_\bull$ is a contravariant functor from
$\Delta$ to the category of graded supermanifolds. (We leave open here
whether we are working in the smooth, analytic or algebraic setting.)
Here, $M_k$ is the value of $M_\bull$ at the object $[k]$, and
$f^*:M_\ell\to M_k$ is the action of the arrow $f:[k]\to[\ell]$ of
$\Delta$. The arrow $d_i:[k]\to[k+1]$ which takes $j<i$ to $j$ and
$j\ge i$ to $j+1$ is known as a face map, while the arrow
$s_i:[k]\to[k-1]$ which takes $j\le i$ to $j$ and $j>i$ to $j-1$ is
known as a degeneracy map.

The simplicial graded supermanifolds used in this paper are the \Cech\
nerves $N_\bull\U$ of covers $\U=\{U_\alpha\}_{\alpha\in I}$. The face
map $\delta_i=d_i^*:N_{k+1}\U\to N_k\U$ corresponds to the inclusion
of the open subspace
\begin{equation*}
  U_{\alpha_0\dots\alpha_{k+1}} \subset N_{k+1}\U
\end{equation*}
into the open subspace
\begin{equation*}
  U_{\alpha_0\dots\widehat{\alpha}_i\dots\alpha_{k+1}} \subset N_k\U ,
\end{equation*}
and the degeneracy map $\sigma_i=s_i^*:N_{k-1}U\to N_kU$ corresponds
to the identification of the open subspace
\begin{equation*}
  U_{\alpha_0\dots\alpha_{k-1}} \subset N_{k-1}\U
\end{equation*}
with the open subspace
\begin{equation*}
  U_{\alpha_0\dots\alpha_i\alpha_i\dots\alpha_k} \subset N_k\U .
\end{equation*}
Any simplicial map $f^*:M_\ell\to M_k$ is the composition of a
sequence of face maps followed by a sequence of degeneracy maps. In
particular, we see that in the case $M_\bull=N_\bull\U$ of the nerve
of a cover, all of these maps are \'etale.

A covariant functor $X^\bull$ from $\Delta$ to a category $\CC$ is
called a cosimplicial object of $\CC$. These arise as the result of
applying a contravariant functor to a simplicial space: for example,
applying the de~Rham functor $\Om^\bull(-)$ to the simplicial graded
supermanifold $N_\bull\U$, we obtain the cosimplicial differential
graded commutative superalgebra $\Om^\bull(N_\bull\U)$. We will also
be interested in the cosimplicial graded Lie superalgebra
\begin{equation*}
  \F(N_\bull\U)
\end{equation*}
with the Batalin--Vilkovisky antibracket, the cosimplicial
differential graded Lie superalgebra
\begin{equation*}
  \B(N_\bull\U)
\end{equation*}
with the Soloviev antibracket and differential $d$, and the
cosimplicial curved graded Lie superalgebra
\begin{equation*}
  \B(N_\bull\U)[[u]]
\end{equation*}
with the Soloviev antibracket, differential $d_u$ and curvature $u\D$.

Associated to a cosimplicial superspace $V^\bull$ is a graded
superspace $N^\bull(V)$, called the normalized cochain complex,
defined as follows:
\begin{equation*}
  N^k(V) = \bigcap_{i=0}^{k-1} \, \ker \bigl( s^i \colon V^k \to
  V^{k-1} \bigr) .
\end{equation*}
This is a complex, with differential
\begin{equation*}
  \delta = \sum_{i=0}^{k+1} (-1)^i d^i : N^k(V) \to N^{k+1}(V) .
\end{equation*}
If the superspaces $V^k$ making up the cosimplicial superspace are
themselves complexes $V^k=V^{\bull k}$, with differential
$d:V^{jk}\to V^{j+1,k}$, we obtain a double complex, with external
differential $\delta:N^k(V^j)\to N^{k+1}(V^j)$ and internal
differential $d:N^k(V^j)\to N^k(V^{j+1})$: the totalization of
$V^{\bull\bull}$ is the graded superspace
\begin{equation*}
  |V|^n = \prod_{k=0}^\infty N^k(V^{n-k}) ,
\end{equation*}
with differential $d_{\Tot}=\delta+(-1)^kd$.

The de~Rham complex of the simplicial graded supermanifold $N_\bull\U$
is the totalization $|\Om^\bull(N_\bull\U)|$ of the de~Rham complex of
$N_\bull\U$. There is a morphism $N_\bull\U\to M$ of simplicial graded
supermanifolds from $N_\bull\U$ to the constant simplicial graded
supermanifold $M$, which induces a quasi-isomorphism of complexes
\begin{equation*}
  \Om^\bull(M) \to |\Om^\bull(N_\bull\U)| .
\end{equation*}
We have $d\nu=\epsilon^*\om$ and $\delta\nu+d\mu=0$. Let
$[\om]\in\Om^0(N_2\U)$ be the \Cech\ differential of
$\mu\in\Om^0(N_1\U)$:
\begin{equation*}
  [\om]_{\alpha_0\alpha_1\alpha_2} =
  \mu_{\alpha_0\alpha_1}|_{U_{\alpha_0\alpha_1\alpha_2}}
  - \mu_{\alpha_0\alpha_2}|_{U_{\alpha_0\alpha_1\alpha_2}}
  + \mu_{\alpha_1\alpha_2}|_{U_{\alpha_0\alpha_1\alpha_2}} .
\end{equation*}
We have $d[\om]=d\delta\mu=\delta d\mu=-\delta^2\mu=0$, hence
$[\om]$ is a locally constant \Cech\ 2-cocycle. By the formula
\begin{equation*}
  d_{\Tot}(\nu+\mu) = \epsilon^*\om - [\om] ,
\end{equation*}
we see that $[\om]$ represents the  cohomology class of $\om$ in
the \Cech\ complex.

The cocycle $[\om]$ is irrelevant in classical mechanics, since being
locally constant, it does not contribute to the Euler-Lagrange
equations. It assumes great importance in quantum mechanics, since it
measures shifts in the phase of the Feynman integrand.

The construction of $|\Om^\bull(N_\bull\U)|$ behaves well under
refinement of covers. A refinement $\V=\{V_\beta\}_{\beta\in J}$ of a
cover $\U=\{U_\alpha\}_{\alpha\in I}$ is determined by a function of
indexing sets $\phi:J\to I$, such that for all $\beta\in J$, $V_\beta$
is a subset of $U_{\phi(\beta)}$. There is a morphism of cosimplicial
differential graded superalgebras
$\Phi^*:\Om^\bull(N_\bull\U)\to\Om^\bull(N_\bull\V)$, obtained by
restricting differential forms on
$U_{\phi(\alpha_0)\dots\phi(\alpha_k)}$ to differential forms on
$V_{\alpha_0\dots\alpha_k}$. Applying the totalization functor, we
obtain a morphism of complexes
\begin{equation*}
  \Phi^* : |\Om^\bull(N_\bull\U)| \to |\Om^\bull(N_\bull\V)| .
\end{equation*}
If we have a further refinement $\W=\{W_\gamma\}_{\gamma\in K}$ of
$\V=\{V_\beta\}_{\beta\in J}$ with $\psi:K\to J$, we may define a
composition of these refinements $\phi\psi:K\to I$, and we obtain a
commuting triangle of morphisms of complexes
\begin{equation*}
  \begin{tikzcd}
    & {}|\Om^\bull(N_\bull\U)| \arrow{dl}[']{\Phi^*}
    \arrow{dr}{\Psi^*\Phi^*} \\
    {}|\Om^\bull(N_\bull\V)| \arrow{rr}[']{\Psi^*} & &
    {}|\Om^\bull(N_\bull\W)|
  \end{tikzcd}
\end{equation*}
In the special case where the cover $\U=\{M\}$ has just one element,
the whole space $M$, we obtain the commutative diagram
\begin{equation*}
  \begin{tikzcd}
    & \Om^\bull(M) \arrow{dl} \arrow{dr} \\
    {}|\Om^\bull(N_\bull\V)| \arrow{rr}[']{\Psi^*} & &
    {}|\Om^\bull(N_\bull\W)|
  \end{tikzcd}
\end{equation*}

In this section, we generalize the classical master equation of
Batalin--Vilkovisky theory to a Maurer--Cartan equation for the
cosimplicial graded Lie superalgebra $\F(N_\bull\U)$. We might expect
this generalization to simply be the Maurer--Cartan equation for the
totalization $|\F(N_\bull\U)|$, but $|\F(N_\bull\U)|$ is not a
differential graded Lie superalgebra. (This problem is related to the
absence of a natural graded commutative product on the singular
cochains of a topological space.) To circumvent this difficulty, we
use a technique introduced in rational homotopy theory by Sullivan
\cite{Sullivan} (see also Bousfield and Guggenheim \cite{BG}), the
Thom--Whitney normalization.

Let $\Om_k$ be the free graded commutative algebra with generators
$t_i$ of degree $0$ and $dt_i$ of degree $1$, and relations
\begin{equation*}
  t_0 + \dotsb + t_k = 1
\end{equation*}
and $dt_0+\dotsb+dt_k=0$. There is a unique differential $\delta$ on
$\Om_k$ such that $\delta(t_i)=dt_i$, and $\delta(dt_i)=0$.

The differential graded commutative algebras $\Om_k$ are the
components of a simplicial differential graded commutative algebra
$\Om_\bull$ (that is, contravariant functor from $\Delta$ to the
category of differential graded commutative algebras): the arrow
$f:[k]\to[\ell]$ in $\Delta$ acts by the formula
\begin{equation*}
  f^*t_i = \sum_{f(j)=i} t_j , \quad 0\le i\le n .
\end{equation*}

The Thom--Whitney normalization of a cosimplicial superspace is an
example of the categorical construction called an end:
\begin{equation*}
  N^\bull_{\TW}(V) = \int_\Delta \Om_\bull \o V^\bull .
\end{equation*}
In other words, $N^\bull_{\TW}(V)$ is the equalizer of the maps
\begin{equation*}
  \begin{tikzcd}
    \displaystyle
    \prod_{k=0}^\infty \Om_k \o V^k
    \arrow[shift right]{r}[']{1\o f_*}
    \arrow[shift left]{r}{f^*\o1} &
    \displaystyle
    \prod_{k,\ell=0}^\infty \prod_{f:[k]\to[\ell]} \Om_k \o V^\ell
  \end{tikzcd}
\end{equation*}
In \cite{Whitney} Whitney defines an injective morphism between the
two normalizations
\begin{equation*}
  \OM : N(V) \to N_{\TW}(V)
\end{equation*}
compatible with the differentials. The Whitney map takes a \Cech\
$k$-cochain $(\nu_{\alpha_0\dots\alpha_k})$ to
\begin{equation*}
  \OM(\nu) = \frac{1}{k+1} \sum_{\alpha_0,\dots,\alpha_k\in I}
  \sum_{i=0}^k (-1)^i
  t_{\alpha_i}dt_{\alpha_0}\dots\widehat{dt}_{\alpha_i}\dots
  dt_{\alpha_k} \o \nu_{\alpha_0\dots\alpha_k} .
\end{equation*}
The differential $\delta\nu$ is taken by this map to
\begin{equation*}
  \OM(\delta\nu) = \frac{1}{k+2} \sum_{\alpha_0,\dots,\alpha_{k+1}\in I}
  \sum_{j=0}^{k+1} \sum_{i=0}^{k+1} (-1)^{i+j}
  t_{\alpha_i}dt_{\alpha_0}\dots\widehat{dt}_{\alpha_i}\dots
  dt_{\alpha_{k+1}} \o \nu_{\alpha_0\dots\widehat{\alpha}_j\dots\alpha_{k+1}} .
\end{equation*}
Only the terms with $i=j$ contribute, and we obtain
\begin{equation*}
  \OM(\delta\nu) = \sum_{\alpha_1,\dots,\alpha_{k+1}\in I}
  dt_{\alpha_1}\dots dt_{\alpha_{k+1}} \, \nu_{\alpha_1\dots\alpha_{k+1}} .
\end{equation*}
On the other hand, we have
\begin{equation*}
  \delta\OM(\nu) = \sum_{\alpha_0,\dots,\alpha_k\in I}
  dt_{\alpha_0}\dots dt_{\alpha_k} \o \nu_{\alpha_0\dots\alpha_k} ,
\end{equation*}
and we conclude that
\begin{equation}
  \label{Whitney}  
  \delta\OM(\nu)=\OM(\delta\nu) .
\end{equation}

If the superspaces $V^k$ making up the cosimplicial superspace are
themselves graded $V^k=V^{\bull k}$, with differential
$d:V^{jk}\to V^{j+1,k}$, we obtain a double complex, with external
differential $\delta:N_{\TW}^k(V^j)\to N_{\TW}^{k+1}(V^j)$ and
internal differential $d:N_{\TW}^k(V^j)\to N_{\TW}^k(V^{j+1})$: the
Thom--Whitney totalization of $V^{\bull\bull}$ is the graded
superspace
\begin{equation*}
  \|V\|^n = \prod_{k=0}^\infty N_{\TW}^k(V^{n-k}) ,
\end{equation*}
with differential $d_{\TW}=\delta+(-1)^kd$. The Whitney map $\OM$
induces an injective morphism of graded superspaces
\begin{equation*}
  \OM : |V|^\bull \to \|V\|^\bull .
\end{equation*}
By \eqref{Whitney}, this is a morphism of complexes. For the
cosimplicial superspaces which we consider in this paper, $\OM$ is a
quasi-isomorphism. (This is proved using a spectral sequence, and we
must impose additional hypotheses in order for the spectral sequence
to converge. It is sufficient to assume that $V^\bull$ is the (graded
super)space of sections of a sheaf over the \Cech\ nerve of a cover
$\U$ of bounded dimension; that is, $U_{\alpha_0\dots\alpha_k}$ is
empty if the cardinality of the set of indices
$\{\alpha_0,\dots,\alpha_k\}$ is sufficiently large. In particular,
this condition holds if the cover is finite.)

Applying this construction to the cosimplicial complex
$\Om^\bull(N_\bull\U)$, we obtain the Thom--Whitney totalization
$\|\Om^\bull(N_\bull\U)\|$, and an injective morphism of complexes
\begin{equation*}
  |\Om^\bull(N_\bull\U)| \hookrightarrow \|\Om^\bull(N_\bull\U)\| .
\end{equation*}
The advantage of Sullivan's Thom--Whitney normalization is that it
takes cosimplicial differential graded commutative superalgebras to
differential graded commutative superalgebras. (Its disadvantage is
that its use is restricted to characteristic zero.) The reason is very
simple: if $V^k$ is a differential graded commutative superalgebra,
then so is $\Om_k \o V^k$. The differential on $\Om_k\o V^k$ is
\begin{equation*}
  d(\alpha\o v) = \delta\alpha \o v + (-1)^i \alpha \o dv , \quad
  \alpha\in\Om^i_k, v\in V^{jk} ,
\end{equation*}
and the product is
\begin{equation*}
  (\alpha_1\o v_1)(\alpha_2\o v_2) = (-1)^{i_2\pa(v_1)} \,
  \alpha_1\alpha_2 \, v_1v_2 ,
\end{equation*}
where $\alpha_\ell\in\Om^{i_\ell}_k$ and $v_\ell\in V^{j_\ell k}$. The
Thom--Whitney totalization $\|V\|$ is a subspace of the product of
differential graded commutative superalgebras $\Om_k\o V^k$, and this
subspace is preserved by the differential and by the product. In this
way, we see that by expanding $|\Om^\bull(N_\bull\U)|$ to the larger
complex $\|\Om^\bull(N_\bull\U)\|$, we obtain a construction which
associates to the cover $\U$ a differential graded commutative
superalgebra.

The Thom--Whitney totalization also takes cosimplicial curved Lie
superalgebras to curved Lie superalgebras. In particular, if $L^\bull$
is a cosimplicial curved Lie superalgebra, the antibracket on $\|L\|$
is given by the formula
\begin{equation*}
  ( \alpha_1\o v_1 , \alpha_2\o v_2 ) = (-1)^{j_2(\pa(v_1)+1)} \,
  \alpha_1\alpha_2 \, ( v_1 , v_2 ) .
\end{equation*}
The Thom--Whitney totalization $\|V\|$ is a subspace of the product
curved Lie superalgebra $\prod_k\Om_k\o V^k$, which is preserved by
the antibracket. In the curved case, the curvatures of the curved Lie
superalgebras $V^k$ assemble to an element of degree $1$ in
$|V|\subset\|V\|$, which is easily seen to be a curvature element for
the Thom--Whitney totalization.

It follows that the Thom--Whitney totalizations $\|\F(N_\bull\U)\|$
and
\begin{equation*}
  \|\B(N_\bull\U)\|
\end{equation*}
are differential graded Lie superalgebras. The differential of
$\|\F(N_\bull\U)\|$ is induced by the differentials $\delta$ of the
algebras $\Om_k$, while the differential of $\|\B(N_\bull\U)\|$ also
involves the internal differential of $\B(N_\bull\U)$. The antibracket
of $\|\F(N_\bull\U)\|$ is induced by the Batalin--Vilkovisky
antibracket on $\F(N_\bull\U)$, while the antibracket of
$\|\B(N_\bull\U))$ is induced by the Soloviev antibracket. Similarly,
the Thom--Whitney totalization
\begin{equation*}
  \|\B(N_\bull\U)[[u]]\| \cong \|\B(N_\bull\U)\|[[u]]
\end{equation*}
is a curved Lie superalgebra, whose differential is the sum the
differential $d_u$ of $\B(N_\bull\U)[[u]]$, with curvature $\D$, and
$\delta$, and whose antibracket is induced by the Soloviev
antibracket.

Given a refinement $\V$ of a cover $\U$, and a refinement $\W$ of
$\V$, we obtain a commuting diagram of Thom--Whitney totalizations
\begin{equation*}
  \begin{tikzcd}
    & \|\F(N_\bull\U)\| \arrow{dl}[']{\Phi^*}
    \arrow{dr}{\Psi^*\Phi^*} \\
    \|\F(N_\bull\V)\| \arrow{rr}[']{\Psi^*} & & \|\F(N_\bull\W)\|
  \end{tikzcd}
\end{equation*}
The arrows in this diagram are morphisms of differential graded Lie
algebras. There are also commuting triangles of differential graded
Lie superalgebras
\begin{equation*}
  \begin{tikzcd}
    & \|\B(N_\bull\U)\| \arrow{dl}[']{\Phi^*}
    \arrow{dr}{\Psi^*\Phi^*} \\
    \|\B(N_\bull\V)\| \arrow{rr}[']{\Psi^*} & &
    \|\B(N_\bull\W)\|
  \end{tikzcd}
\end{equation*}
and of curved Lie superalgebras
\begin{equation*}
  \begin{tikzcd}
    & \|\B(N_\bull\U)[[u]]\| \arrow{dl}[']{\Phi^*}
    \arrow{dr}{\Psi^*\Phi^*} \\
    \|\B(N_\bull\V)[[u]]\| \arrow{rr}[']{\Psi^*} & &
    \|\B(N_\bull\W)[[u]]\|
  \end{tikzcd}
\end{equation*}

The analogue of the classical master equation \eqref{MASTER} in the
global setting is the Maurer--Cartan equation for the differential
graded superalgebra $\|\B(N_\bull\U)\|$:
\begin{equation*}
  d_{\TW}\S + \half \[ \S , \S \] = 0 .
\end{equation*}
Here, $\S$ is a collection of elements
$\S^j_{\alpha_0\dots\alpha_k} \in \Om^j_k \o
\B^{-j}(U_{\alpha_0\dots\alpha_k})$ of total degree $0$, simplicial,
in the sense that for each $f:[k]\to[\ell]$,
\begin{equation*}
  (f^*\o1)\S^j_\ell=(1\o f_*)\S^j_k ,  
\end{equation*}
which satisfies the sequence of Maurer--Cartan equations
\begin{equation*}
  \delta\S^{j-1} + (-1)^j d\S^j + \frac12 \sum_{i=0}^j \[ \S^i ,
    \S^{j-i} \] = 0 .
\end{equation*}
This makes the following definition natural. The graded derivation on
the curved Lie superalgebra $\|\B(N_\bull\U)[[u]]\|$ is
\begin{equation*}
  d_{\TW,u}=d_{\TW}+u\iota .
\end{equation*}
\begin{definition}
  Let $M$ be a graded supermanifold $M$. A global covariant field theory for
  $M$ is a solution of the Maurer--Cartan equation for the curved Lie
  superalgebra $\|\B(N_\bull\U)[[u]]\|$, where $\U$ is a cover of $M$:
  \begin{equation*}
    d_{\TW,u}\S_u + \half \[ \S_u , \S_u \] = - u \D .
  \end{equation*}
\end{definition}

Since the Maurer--Cartan set $\MC(L)$ is a functor on the category of
curved Lie superalgebras, there is a commutative diagram of sets
\begin{equation*}
  \begin{tikzcd}[cramped,column sep=tiny]
 & \MC(\|\B(N_\bull\U)[[u]]\|) \arrow{dl}[']{\Phi^*}
    \arrow{dr}{\Psi^*\Phi^*} \\
    \MC(\|\B(N_\bull\V)[[u]]\|) \arrow{rr}[']{\Psi^*} & &
    \MC(\|\B(N_\bull\W)[[u]]\|)
  \end{tikzcd}
\end{equation*}
We see that if $\S_u$ is a covariant field theory with respect to a
cover $\U$ of $M$, then it induces a covariant field theory with
respect to any refinement of $\U$.

We now come to the main result of this section.
\begin{theorem}
  \label{global}
  Let $M$ be a graded supermanifold with symplectic form $\om$. Let
  $\U$ be a cover of $M$, and let $(\nu,\mu)\in|\Om^\bull(N_\bull\U)|$
  be a one-form such that $d\nu=\epsilon^*\om$ and $\delta\nu=d\mu$.

  Let $\S_u\in\check{C}^0(\U,\B[[u]])$ be the \Cech\ cochain which
  over $U_\alpha$ equals the local covariant field theory
  $\S_{\alpha,u}$ associated to the one-form $\nu_\alpha$.

  The element $\SS_u=\OM(\S_u+\mu\eps)$ is a global covariant field
  theory, that is, a Maurer--Cartan element in the curved Lie
  superalgebra $\|\B(N_\bull\U)[[u]]\|$.
\end{theorem}

This result is proved by lengthy calculation. One subtle point is that
\begin{equation*}
  \delta\OM(\mu\eps) = \OM(\delta\mu\eps) \in
  \check{C}^2(\U,\B[[u]])
\end{equation*}
vanishes. Indeed,
\begin{equation*}
  \delta\mu \in \check{C}^2(\U,\CO)
\end{equation*}
is locally constant and by definition constant multiples of $\eps$
vanish in the sheaf $\B$.

Over the set $U_{\alpha_0\dots\alpha_k}$, the covariant field theory $\SS_u$ is
given by the explicit formula
\begin{equation*}
  \SS_u
  = \sum_{i=0}^k t_{\alpha_i} \o \S_{\alpha_i,u} \bigm|_{U_{\alpha_0\dots\alpha_k}}
  + \frac{1}{2} \sum_{i,j=0}^k
  (t_{\alpha_i}dt_{\alpha_j}-t_{\alpha_j}dt_{\alpha_i}) \o
  \mu_{\alpha_i\alpha_j} \bigm|_{U_{\alpha_0\dots\alpha_k}}\eps .
\end{equation*}
Let $W\in\CO(M)$ be a function on $M$ of degree $1$ and odd parity
such that $\{W,W\}=0$. As in the local case, we may twist this global
covariant field theory by $u^{-1}W$:
\begin{align*}
  \SS_u\bull u^{-1}W
  &= \SS_u + u^{-1} \bigl( d_uW + \[ \SS_u , W \] \bigr) \\
  &= \SS_u +  W\eps + u^{-1} \sum_{i=0}^k t_{\alpha_i} \o \[
    \S_{\alpha_i,u} , W \] \bigr) \bigm|_{U_{\alpha_0\dots\alpha_k}} .
\end{align*}
More generally, $W$ might be any \Cech\ cocycle
$W\in\check{C}^\bull(\U,\CO)$ of degree $1$ and odd parity, such that
$\{\OM(W),\OM(W)\}=0$.

We close this section with a discussion of how the global covariant
field theory associated to the class $(\nu,\mu)$ changes under an
equivalence of theories. The type of equivalence we have in mind is a
homotopy of the form $\tilde{\nu}\in\check{C}^0(\U,\Om^0)$ between a
pair of classes $(\nu_i,\mu_i)$, $i=0,1$:
\begin{align*}
  \nu_1 - \nu_0 &= d\tilde{\nu} & \mu_1 - \mu_0 &= \delta\tilde{\nu} .
\end{align*}
\begin{proposition}
  Let $\SS_{i,u}$ be the global covariant field theory associated to
  $(\nu_i,\mu_i)$. Then $\SS_{0,u}$ and $\SS_{1,u}$ are gauge
  equivalent:
  \begin{equation*}
    \SS_{0,u}\bull\OM(\tilde{\nu}\eps) = \SS_{1,u} .
  \end{equation*}
\end{proposition}
\begin{proof}
  We have
  \begin{align*}
    d_{\TW,u}\OM(\tilde{\nu}\eps)
    &= \OM(\delta\tilde{\nu}\eps) + \p \OM(\tilde{\nu}) \\
    &= \OM( \mu_1 - \mu_0 ) + (-1)^{\pa(\xi^a)} \OM( ( \nu_{1,a} -
      \nu_{0,a} ) \p\xi^a )
  \end{align*}
  and, for $i=0,1$,
  \begin{align*}
    \[ \SS_{i,u} , \OM(\tilde{\nu}\eps) \]
    &= - u \OM( \xi^+_a \pi^{ab} \p_b\tilde{\nu}\eps) \\
    &= - u \OM( \xi^+_a \pi^{ab} (\nu_{1,b}-\nu_{0,b})\eps ) .
  \end{align*}
  Adding these two equations, we see that
  \begin{equation*}
    d_{\TW,u}\OM(\tilde{\nu}\eps) + \[ \SS_{i,u} ,
      \OM(\tilde{\nu}\eps) \] = \SS_{1,u} - \SS_{0,u} .
  \end{equation*}
  Taking the difference of these two equations, we see that
  \begin{equation*}
    \[ \SS_{1,u} - \SS_{0,u} , \OM(\tilde{\nu}\eps)) \] = 0 .
  \end{equation*}
  We see that
  \begin{equation*}
    \SS_{0,u}\bull\OM(\tilde{\nu}\eps) = \SS_{0,u} + \bigl(
    \SS_{1,u} - \SS_{0,u} \bigr) = \SS_{1,u} ,
  \end{equation*}
  proving the result.
\end{proof}

It follows from this proposition that the global covariant field
theory $\SS_u$ is invariantly associated, up to refinement of the
cover $\U$ over which it is defined and a gauge transformation, to an
element of the Deligne cohomology group
\begin{equation*}
  \check{\mathbb{H}}^2(\U,\R \to \Om^0 \xrightarrow{d} \Om^1) .
\end{equation*}
Quantization of this model requires lifting the two-cocycle
$\delta\mu\in\check{C}^2(\U,\R)$ to $\check{C}^2(\U,\Z)$. With this
constraint, the global covariant field theory is classified by an
element of the Deligne cohomology group
\begin{equation*}
  \check{\mathbb{H}}^2(\U,\Z \to \Om^0 \xrightarrow{d} \Om^1)
\end{equation*}
which classifies Hermitian line bundles with connection on $M$ which
are trivialized on restriction to the cover $\U$.

\section{The particle as a covariant field theory}

In this section, we couple certain covariant field theories to gravity
on the world-line. Of course, one-dimensional gravity carries no
propagating fields: instead, the effect of coupling to gravity is to
render the covariant field theory generally covariant.

The gravitational field in one dimension consists of a nowhere
vanishing one-form $e$ on the world-line, whose square is the metric
along the world-line. Its antifield $e^+$ is a fermionic scalar field
of ghost number $-1$. In addition, there is a ghost field $c$, which
is a fermionic field that transforms as a vector field on the
world-line, and has ghost number $1$: its antifield $c^+$ is a bosonic
field that transforms as a quadratic differential on the world-line,
and has ghost number $-2$.

Consider the graded manifold $T^*\R[1]\cong\R[1]\times\R[-1]$, with
fermionic coordinates $b$ and $c$, respectively of ghost number $-1$
and $1$. We consider the covariant field theory $\X_u$ associated to
the one-form $\nu=-c\,db$, given by the explicit formula
\begin{equation*}
  \X_u = \X_0 + u\X_1 = c\p b + u(b^+c^++c^+c\eps) .
\end{equation*}
Consider the Batalin--Vilkovisky Hamiltonian flow generated by the
Hamiltonian
\begin{equation*}
  \log(b^+) c^+c ,
\end{equation*}
defined in a neighbourhood of the locus where $b^+=1$. The covariant
field theory $\X_u$ flows to
\begin{equation*}
  \X_u\bull \tau\log(b^+)c^+c = (b^+)^{\tau-1} c ( b^+\p b + \tau
  c^+\p c ) + u (b^+)^{\tau-1}c^+ + (1-\tau)uc^+c\eps ,
\end{equation*}
and, setting $\tau=1$, we see that
\begin{equation*}
  \X_u\bull\log(b^+)c^+c = c(b^+\p b+c^+\p c) + uc^+ .
\end{equation*}
We may identify $b^+$ as the gravitational field $e$. The antifield
$e^+$ of $e$ is the field $-b$, and the action in these coordinates
becomes
\begin{equation*}
  \X_u\bull\log(b^+)c^+c = c(-e\p e^++c^+\p c) + uc^+ .
\end{equation*}

Let $\S_u=\S_0+u\S_1$ be a covariant field theory with $\S_i=0$,
$i>1$; denote by $u\D$ its curvature. The product of the covariant
field theories $\S_u$ and $\X_u$ is associated to the symplectic
graded manifold $M\times T^*\R[1]$:
\begin{equation*}
  \S_u + \X_u = \S_0 + c\p b + u(\S_1 + b^+c^+ + c^+c\eps ) .
\end{equation*}
The following theorem shows that after a further gauge transformation,
generated by the Hamiltonian $c\S_1$, this model is transformed into a
theory minimally coupled to the background gravitational field.
\begin{theorem}
  \label{particle}
  Let $M$ be a graded supermanifold, and let $\nu$ be a one-form on
  $M$ such that $d\nu$ is a symplectic form. Let $\S_u$ be the
  associated covariant field theory. Then we have
  \begin{equation*}
    ( \S_u + \X_u )\bull\log(b^+)c^+c\bull c\S_1 = \S_0
    + c(\D+b^+\p b+c^+\p c) + c\iota\S_0 + uc^+ .
  \end{equation*}
\end{theorem}

\begin{corollary}
  \label{particleV}
  Let $V\in\CO(M)$ be a function on $M$ of ghost number $0$ and even
  parity, and let $W=cV$. After successively applying the gauge
  transformations generated by $\log(b^+)c^+c$ and $c\S_1$, the
  twisted covariant field theory
  \begin{equation*}
    ( \S_u + \X_u )\bull u^{-1}W    
  \end{equation*}
  is transformed into the covariant field theory
  \begin{equation*}
    ( \S_u + \X_u )\bull u^{-1}W\bull\log(b^+)c^+c\bull c\S_1
    = \S_0 - b^+V + c(\D+b^+\p b+c^+\p c) + c\iota\S_0 + uc^+ .
  \end{equation*}
\end{corollary}
\begin{proof}
  We have $d_u(u^{-1}W)=cV\eps$,
  \begin{equation*}
    \[ \S_u , u^{-1}W \] = \[ \S_1 , cV \] = - \[ c\S_1 , V \]
  \end{equation*}
  and
  \begin{equation*}
    \[ \X_u , u^{-1}W \] = - b^+V - cV\eps .
  \end{equation*}
  It follows that
  \begin{equation*}
    ( \S_u + \X_u )\bull u^{-1}W
    = \S_u + \X_u - b^+V - \[ c\S_1 , V \]
  \end{equation*}
  and that
  \begin{multline*}
    ( \S_u + \X_u )\bull u^{-1}W\bull\log(b^+)c^+c
    = \bigl( \S_u + \X_u - b^+V - \[ c\S_1 , V \]
    \bigr)\bull\log(b^+)c^+c \\
    = \S_u + c(b^+\p b+c^+\p c) + uc^+ - b^+V - \[ c\S_1 , b^+V \] .
 \end{multline*}
 We see that
  \begin{align*}
    ( \S_u + \X_u )\bull u^{-1}W\bull\log(b^+)c^+c\bull c\S_1
    &= ( \S_u + c(b^+\p b+c^+\p c) + uc^+ )\bull c\S_1 \\
    & - e^{-\ad(c\S_1)} ( b^+V + \[ c\S_1 , b^+V \] ) .
  \end{align*}
  Since $b^+V + \[ c\S_1 , b^+C \]=e^{\ad(c\S_1)}b^+V$, the corollary
  follows.
\end{proof}

\begin{remark}
  Theorem~\ref{particle} generalizes to the global case without any
  difficulties: if $\SS_u$ satisfies the hypotheses of
  Theorem~\ref{global}, we have
  \begin{equation*}
    ( \SS_u + \X_u )\bull\log(b^+)c^+c\bull c\SS_1 = \SS_0
    + c(\D+b^+\p b+c^+\p c) + c\iota\SS_0 + uc^+ .
  \end{equation*}
\end{remark}

\begin{remark}
  After coupling to gravity, the covariant field theory $\S_u$, which
  is only defined if the two-form $\om=d\nu$ is symplectic, is seen to
  be equivalent to a covariant field theory which is defined for any
  one-form $\nu$ on $M$, without any condition that $d\nu$ is
  nondegenerate.
\end{remark}

Theorem \ref{particle} may be restated in the following suggestive
way.
\begin{proposition}
  We have
  \begin{equation*}
    ( \S_u + \X_u )\bull\log(b^+)c^+c\bull c\S_1
    = ( \S_u + \X_u )\bull(\log(b^+)c^+c\ast c\S_1) ,
  \end{equation*}
  where
  \begin{equation*}
    \log(b^+)c^+c\ast c\S_1 = \frac{\log(b^+)}{b^+-1} \bigl(
    c(\S_1+\X_1) - c^+c \bigr) .
  \end{equation*}
\end{proposition}
\begin{proof}
  Suppose that for all $n\ge0$, we have
  \begin{equation}
    \label{zyz}
    \ad(z)\ad(y)^nz = 0 .
  \end{equation}
  It follows that
  \begin{equation*}
    e^{-t\ad(z)}e^{-\ad(y)}z = e^{-\ad(y)}z ,
  \end{equation*}
  and hence that
  \begin{equation*}
    \ad(y*tz)z = \ad(y)z .
  \end{equation*}

  By Proposition \ref{BCH}, we see that
  \begin{align*}
    y\ast z &= y + \int_0^1 \frac{\ad(y*tz)}
              {1-e^{-t\ad(z)}e^{-\ad(y)}} z \, dt \\
            &= y + \int_0^1 \frac{\ad(y)}{1-e^{-\ad(y)}} z \, dt \\
            &= y + \frac{\ad(y)}{1-e^{-\ad(y)}} z .
  \end{align*}
  
  Let $y=\log(b^+)c^+c$ and $z=c\S_1$. We have
  \begin{equation*}
    \ad(\log(b^+)c^+c)^n c\S_1 = (-\log(b^+))^nc\S_1 ,
  \end{equation*}
  and the hypothesis \eqref{zyz} is satisfied. Thus, we have
  \begin{align*}
    \log(b^+)c^+c\ast c\S_1
    &= \log(b^+)c^+c + \frac{\ad(\log(b^+)c^+c)}{e^{\ad(\log(b^+)c^+c)}-1}
      c\S_1 \\
    &= \log(b^+)c^+c + \frac{\log(b^+)c\S_1}{b^+-1} ,
  \end{align*}
  and the lemma follows.
\end{proof}

Using Corollary \ref{particleV}, we can generalize the discussion of
the particle in a flat spacetime discussed in the introduction,
allowing a curved target with a magnetic field. Let $U$ be an open
subset of $\R^n$, with coordinates $\{x^\mu\}_{1\le\mu\le n}$. Let
\begin{equation*}
  g = g_{\mu\nu}(x) dx^\mu \o dx^\nu
\end{equation*}
be a pseudo-Riemannian metric on $U$, and let $A=A_\mu dx^\mu$ be an
electromagnetic potential on $U$. Let $M=T^*U$ be the manifold with
coordinates $\{x^\mu\}$ and $\{p_\mu\}$. We consider the covariant
field theory $\S_u$ associated to the one-form
\begin{equation*}
  \nu = (p_\mu+A_\mu)dx^\mu .
\end{equation*}
The corresponding symplectic form $\om=d\nu$ equals
\begin{equation*}
  \om = dp_\mu \, dx^\mu + \half F_{\mu\nu}(x) dx^\mu \, dx^\nu ,
\end{equation*}
where
\begin{equation*}
  F_{\mu\nu}(x) = \p_\mu A_\nu(x) - \p_\nu A_\mu(x)
\end{equation*}
is the electromagnetic field. The associated Poisson bracket is
\begin{equation*}
  \{ f , g \} = \frac{\p f}{\p x^\mu} \frac{\p g}{\p p_\mu} - \frac{\p
    f}{\p p_\mu} \frac{\p g}{\p x^\mu} + F_{\mu\nu}(x) \frac{\p f}{\p
    p_\mu} \frac{\p g}{\p p_\nu} ,
\end{equation*}
and the covariant field theory, after coupling to gravity, equals
\begin{align*}
  \S_u + \X_u
  &= ( p_\mu + A_\mu ) \p x^\mu + b\p c \\
  &+ u \bigl( x^+_\mu p^{+\mu} + \half F_{\mu\nu} p^{+\mu} p^{+\nu} +
    b^+c^+ \bigr) + u \bigl( (p_\mu+A_\mu)p^{+\mu} + c^+c \bigr) \eps .
\end{align*}
The field theory describing the particle is obtained by twisting this
field theory by $u^{-1}cV$, where
\begin{equation*}
  V = \half g^{\mu\nu} p_\mu p_\nu .
\end{equation*}

The proof of Theorem~\ref{particle} occupies the remainder of the
section. We use the following formulas, whose proofs are similar to
the proof that $\S_u$ is a covariant field theory:
\begin{gather}
  \label{c}
  d_u(c\S_1) + \[ \S_u , c\S_1 \] = c ( \D + \iota\S_u )
  \intertext{and}
  \label{cc}
  \[ d_u(c\S_1) + \[ \S_u , c\S_1 \] , c\S_1 \] = 2c\p c .
\end{gather}

Let us calculate $(\S_u+c(b^+\p b+c^+\p c)+uc^+)\bull\tau c\S_1$. We
have
\begin{multline*}
  ( \S_u + c(b^+\p b+c^+\p c)+uc^+ )\bull\tau c\S_1 \\
  = \S_u + c(b^+\p b + c^+\p c)+uc^+ + \sum_{n=0}^\infty
  \frac{(-\tau)^{n+1}}{(n+1)!} \ad(c\S_1)^n \d(c\S_1) ,
\end{multline*}
where
\begin{equation*}
  \d(c\S_1) = d_u(c\S_1) + \[ \S_u + c(b^+\p b+c^+\p c)+uc^+ , c\S_1 \] .
\end{equation*}
It follows from \eqref{c} that
\begin{equation*}
  \d(c\S_1) = c ( \D + \iota\S_u ) - c\p c\S_1 - u \S_1 .
\end{equation*}
Since $c^2=0$, we see that both $\[c\p c\S_1,c\S_1\]$ and
$\[ c\iota\S_1 , c\S_1 \]$ vanish. By \eqref{cc}, we see that
\begin{align*}
  \[ \d(c\S_1) , c\S_1 \] & = \[ c ( \D + \iota\S_u ) , c\S_1 \] -
    u \[ \S_1 , c\S_1 \] \\
  & = 2c\p c\S_1 - 2uc \iota\S_1 .
\end{align*}
It is clear that
\begin{equation*}
  \[ \[ \d(c\S_1) , c\S_1 \] , c\S_1 \] = 0 .
\end{equation*}
In summary, we have
\begin{align*}
  ( \S_u + c(b^+\p b+c^+\p c)+uc^+ )\bull\tau c\S_1
    &= \S_0 + c(b^+\p b+c^+\p c) + u ( \S_1 + c^+ ) \\
    &+ \tau \bigl( c ( \D + \iota\S_u ) - c\p c\S_1 - u \S_1
    \bigr) + \frac{\tau^2}{2} \bigl( 2c\p c\S_1 - 2uc \iota\S_1 \bigr) \\
    &= \S_0 + c(\tau\D+b^+\p b+c^+\p c) + \tau c\iota\S_0 + uc^+ \\
    &+ (1-\tau) \Bigl( u\S_1 + \tau  uc \iota\S_1 - \tau c\p c\S_1
    \Bigr) .
\end{align*}
Setting $\tau=1$, we obtain Theorem~\ref{particle}.

\section{The spinning particle as a covariant field theory}

In this section, we study a supersymmetric version of the results of
the last section. Supergravity in one dimension has as its fields the
graviton $e$ and a fermionic field $\chi$, the gravitino, which like
$e$ is transforms as a one-form on the world-line. In addition to the
ghost $c$, there is also a superghost $\gamma$, which is a bosonic
field of ghost number $1$ that transforms as a world-line scalar.

Consider the graded manifold
$T^*\Pi\R[1]\cong\Pi\R[1]\times\Pi\R[-1]$, with bosonic coordinates
$\beta$ and $\gamma$, respectively of ghost number $-1$ and $1$.
Consider the covariant field theory $\Xi_u$ associated to the one-form
$\nu=\gamma\,d\beta$:
\begin{equation*}
  \Xi_u = \Xi_0 + u\Xi_1 = \gamma\p\beta +
  u(\beta^+\gamma^++\gamma^+\gamma\eps) .
\end{equation*}

Theorem~\ref{particle}, applied to the graded supermanifold
$M\times T^*\Pi\R[1]$, shows that
\begin{multline*}
  ( \S_u + \Xi_u + \X_u )\bull\log(b^+)c^+c\bull c(\S_1+\Xi_1)
  = \S_0 \\
  + c(\D+\beta^+\p\beta+\gamma^+\p\gamma+b^+\p b+c^+\p c +
  \iota(\S_0+\gamma\p\beta) ) + \gamma\p\beta + uc^+ .
\end{multline*}
In order to obtain the supersymmetric analogue of the particle, called
the spinning particle, we choose a function $Q\in\CO(M)$ of ghost
number $0$ and odd parity. We twist the covariant field theory
$\S_u+\Xi_u+\X_u$ by an element $u^{-1}W$, where $W$ is given by the
formula (cf.\ Cattaneo and Schiavina \cite{CS}*{Section 6.3}, and
\cite{curved}*{Eq.~(5)})
\begin{equation}
  \label{supertwist}
  W = \half c \{ Q , Q \} + \gamma Q - b\gamma^2 .
\end{equation}
The term linear in $b$ is chosen in such a way as to guarantee that
$\{W,W\}=0$.

Applying the twist $u^{-1}W$ to the covariant field theory
$\S_u+\Xi_u+\X_u$ gives the twisted covariant field theory
\begin{align*}
  ( \S_u + \Xi_u + \X_u )\bull u^{-1}W
  &= \S_u + \Xi_u + \X_u \\
  &- \half b^+\{Q,Q\} + \half \[ \{ Q , Q \} , c\S_1 \] - \beta^+ Q +
    \gamma \[\S_1,Q\] \\
  &+ c^+\gamma^2 - 2b\beta^+\gamma + b\gamma^2\eps .
\end{align*}
Gauging by $\log(b^+)c^+c$ followed by $c\Xi_1$, we obtain
\begin{multline*}
  ( \S_u + \Xi_u + \X_u )\bull u^{-1}W\bull\log(b^+)c^+c\bull c\Xi_1 \\
  \begin{aligned}
    &= \S_u + c(\beta^+\p\beta+\gamma^+\p\gamma+b^+\p b+c^+\p c
    -\gamma\p\beta\eps) + uc^+ \\
    &- \half b^+\{Q,Q\} + \half \[ b^+ \{ Q , Q \} , c\S_1 \] -
    \beta^+ Q + ( \gamma + c\beta^+ - c\gamma\eps) \[\S_1,Q\] \\
    &+ (b^+)^{-1} ( c^+ - \Xi_1) \gamma^2 - 2b\beta^+\gamma -
    (b^+)^{-1} ( c^+ - \Xi_1) c \gamma^2\eps .
  \end{aligned}
\end{multline*}
Gauging by $c\S_1$, and observing that
$c\Xi_1\ast c\S_1=c(\S_1+\Xi_1)$, we obtain
\begin{multline*}
  ( \S_u + \Xi_u + \X_u )\bull u^{-1}W\bull\log(b^+)c^+c\bull
  c(\S_1+\Xi_1) \\
  \begin{aligned}
    &= \S_0 + c(\D+\beta^+\p\beta+\gamma^+\p\gamma+b^+\p b+c^+\p c +
    \iota\S_0-\gamma\p\beta\eps) + uc^+ \\
    &- \half b^+\{Q,Q\} - \beta^+ Q + \gamma \[ \S_1 , Q \] \\
    &+ (b^+)^{-1} ( c^+ - \S_1 - \Xi_1) \gamma^2 (1-c\eps) -
    2b\beta^+\gamma .
  \end{aligned}
\end{multline*}
Substituting the graviton field $e$ for $b^+$ and the gravitino field
$\chi$ for $\beta^+$, and projecting to $\F$, we obtain the action of
the spinning particle:
\begin{multline}
  \label{spinningparticle}
  \tint ( S_u + \Xi_u + \X_u )\bull u^{-1}W\bull\log(e)c^+c\bull
  c(S_1+\chi\gamma^+) \\
  \begin{aligned}
    &= \tint \bigl( (-1)^{\pa(\xi^a)} \nu_a\p\xi^a - \half e\{Q,Q\} -
    \chi Q
    \\ &
    + c(\D-\chi\p\chi^++\gamma^+\p\gamma-e\p e^++c^+\p c)  \\
    &+ \gamma \bigl( - \xi^+_a \pi^{ab} \p_b Q + 2e^+\chi )
    + e^{-1} \gamma^2 ( c^+ + \gamma\p\chi^+ - \half \xi^+_a \pi^{ab}
    \xi^+_b ) \bigr) .
  \end{aligned}
\end{multline}

Solutions of the classical master equation \eqref{master} are
classified by their rank, that is, their degree as functions of the
antifields. Gauge theories whose symmetries close off-shell have rank
1, but the action of the spinning particle \eqref{spinningparticle}
has rank~2, owing to the presence of term
\begin{equation*}
  - \half \tint e^{-1} \gamma^2 \xi^+_a \pi^{ab} \xi^+_b .
\end{equation*}
It is possible that by adjoining auxilliary fields, this covariant
field theory may be shown to be equivalent to a covariant field theory
whose action is of rank 1.

We close this section by showing how, as an application of the above
formulas, the spinning particle may be expressed as a covariant field
theory. We discuss only the local case, leaving its globalization,
which uses Theorem~\ref{global}, to the reader.

We work, as in the last section, with an open subset $U$ of $\R^n$,
with coordinates $\{x^\mu\}$, pseudo-Riemannian metric
\begin{equation*}
  g = g_{\mu\nu}(x) dx^\mu \o dx^\nu ,
\end{equation*}
and electromagnetic field $A=A_\mu\,dx^\mu$.

Denote the basis of $V=\R^n$ by $\{e_a\}$, and let
$\eta_{ab}=\eta(e_a,e_b)$ be an inner product on $V$, of the same
signature as the metric $g_{\mu\nu}$ on $U$. Let
\begin{equation*}
  \theta^a = \theta_\mu^a\,dx^\mu \in \Om^1(U,V)
\end{equation*}
be a moving frame, that is, a one-form defining an isometry between
$T_xU$ and $V$ at each point $x\in U$, so that
\begin{equation*}
  g_{\mu\nu} = \eta_{ab} \theta_\mu^a \theta_\nu^b .
\end{equation*}
Denote by $\theta^\mu_a$ the inverse of $\theta_\mu^a$, in the sense
that
\begin{equation*}
  \theta_\mu^a\theta^\mu_b=\delta^a_b .
\end{equation*}

The connection one-form
$\om^a{}_b=\om_\mu{}^a{}_b\,dx^\mu\in\Omega^1(U,\End(V))$ is the
antisymmetric matrix of one-forms on $U$ characterized by two
conditions: it is compatible with the metric $\eta$,
\begin{equation*}
  \om^b{}_a = - \eta_{a\tilde{a}} \, \eta^{b\tilde{b}}
  \om^{\tilde{a}}{}_{\tilde{b}} ,
\end{equation*}
and satisfies the first Cartan structure equation
\begin{equation*}
  d\theta^a + \om^a{}_b \, \theta^b = 0 .
\end{equation*}

Let $M$ be the supermanifold $T^*U\times\Pi V$, with coordinates
$(x^\mu,p_\mu,\psi^a)$, and consider the one-form
\begin{equation*}
  \nu = (p_\mu+A_\mu)dx^\mu + \half \eta_{ab} \psi^a d\psi^b .
\end{equation*}
The symplectic form $d\nu$ on $M$ equals
\begin{equation*}
  d\nu = dp_\mu\,dx^\mu + \half F_{\mu\nu}dx^\mu dx^\nu + \half \eta_{ab}
  d\psi^a d\psi^b .
\end{equation*}
The associated Poisson bracket is
\begin{equation*}
  \{ f , g \} = \frac{\p f}{\p x^\mu} \frac{\p g}{\p p_\mu} - \frac{\p
    f}{\p p_\mu} \frac{\p g}{\p x^\mu} + F_{\mu\nu}(x) \frac{\p f}{\p
    p_\mu} \frac{\p g}{\p p_\nu} + (-1)^{\pa(f)} \eta^{ab}
  \frac{\p f}{\p\psi^a} \frac{\p g}{\p\psi^b} ,
\end{equation*}
and the covariant field theory, after coupling to supergravity, equals
\begin{multline*}
  \S_u + \Xi_u + \X_u = ( p_\mu + A_\mu ) \p x^\mu - \half \eta_{ab}
  \psi^a \p\psi^b + \beta\p\gamma + b\p c \\
  \begin{aligned}
    &+ u \bigl( x^+_\mu p^{+\mu} + \half F_{\mu\nu} p^{+\mu} p^{+\nu}
    + \half \eta^{ab} \psi^+_a\psi^+_b + \beta^+\gamma^+ + b^+c^+ \bigr) \\
    &+ u \bigl( (p_\mu+A_\mu)p^{+\mu} + c^+c \bigr) \eps .
  \end{aligned}
\end{multline*}
Let
\begin{equation*}
  \tilde{p}{}_\mu = p_\mu + \half \om_{\mu a b} \psi^a\psi^b .
\end{equation*}
The spinning particle in a curved background \cite{curved} is the
twisted covariant field theory
\begin{equation*}
  (\S_u+\Xi_u+\X_u)\bull u^{-1}(\half c\{Q,Q\}+\gamma Q+b\gamma^2) ,
\end{equation*}
where $Q = \theta^\mu_a \psi^a \tilde{p}{}_\mu$. Observe that the
quantization of $Q$ is is the Dirac operator on $U$. The proof of
Lichnerowicz's formula for the square of the Dirac operator shows that
\begin{equation*}
  \{ \D , \D \} = \theta^\mu_a(x) \theta^\nu_b(x) \bigl( \eta^{ab}
  \tilde{p}{}_\mu \tilde{p}{}_\nu - \half F_{\mu\nu}(x) \psi^a \psi^b
  \bigr) .
\end{equation*}

  \bigskip
  \paragraph*{Acknowledgements} {\small \setstretch{1.05}
  I am grateful to Chris Hull for introducing me to the first-order
  formalism of the spinning particle, and to Si Li, Pavel Mn\"ev and
  Sean Pohorence for further insights. This research is partially
  supported by EPSRC Programme Grant EP/K034456/1 ``New Geometric
  Structures from String Theory,'' a Fellowship of the Simons
  Foundation, and Collaboration Grants \#243025 and \#524522 of the
  Simons Foundation. Parts of this paper were written while the author
  was visiting the Yau Mathematical Sciences Center at Tsinghua
  University and the Department of Mathematics of Columbia University,
  as a guest of Si Li and Mohammed Abouzaid respectively.

}


\end{document}